\documentclass{IEEEojcsys}

\usepackage{amsmath,amssymb,amsthm}
\usepackage{mathrsfs}
\usepackage{booktabs}
\usepackage{graphicx}
\usepackage{etoolbox}
\usepackage{cite}
\usepackage{listings}
\usepackage{tikz}
\usepackage{pgfplots}
\pgfplotsset{compat=1.16}
\pgfplotsset{
    masubuchi axis/.style={
        width=0.90\columnwidth,
        height=0.70\columnwidth,
        font=\footnotesize,
        ymin=5.7,
        ymax=7.7,
        ytick={5.8,6.2,6.6,7.0,7.4},
        ylabel={solver-returned objective $ \gamma $},
        grid=major,
        grid style={gray!22}
    }
}
\usetikzlibrary{calc,plotmarks,positioning}
\usepackage{xcolor}
\usepackage{url}
\usepackage{bm}
\usepackage[hidelinks]{hyperref}
\usepackage{cleveref}

\makeatletter
\patchcmd{\@maketitle}
  {{\receivedfont \ifx\@receiveddate\@empty\else Received\ \@receiveddate\fi\ifx\@revised\@empty\else;\ revised\ \@reviseddate\fi\ifx\@accepteddate\@empty\else;\ accepted\ \@accepteddate\fi\ifx\@publisheddate\@empty\else; Date of publication\space\@publisheddate\fi\ifx\@currentdate\@empty\else;\space date of current version\space\@currentdate\fi.\space \@editor\par}}
  {}
  {}
  {\PackageError{paper}{Unable to suppress the publisher metadata line}{The OJ-CSYS class title implementation may have changed.}}
\def\submissionpagefooters{%
  \def\@evenfoot{\hbox to \textwidth{{\rfxfont\thepage}\hfil}}%
  \def\@oddfoot{\hbox to \textwidth{\hfil{\rfxfont\thepage}}}%
}
\let\submissionvendorheadings\ps@headings
\def\ps@headings{%
  \submissionvendorheadings
  \submissionpagefooters
}
\let\submissionvendorplain\ps@plain
\def\ps@plain{%
  \submissionvendorplain
  \submissionpagefooters
}
\submissionpagefooters
\makeatother

\newtheorem{lemma}{Lemma}
\graphicspath{{pics/}{biographAndphoto/}}

\theoremstyle{plain}

\newtheorem{assumption}{Assumption}

\newtheorem{remark}{Remark}
\newtheorem{corollary}{Corollary}

\newtheorem{example}{Example}
\crefname{equation}{}{}
\crefname{algorithm}{Algorithm}{Algorithms}
\crefname{figure}{Fig.}{Figs.}
\crefname{proposition}{Proposition}{Propositions}
\crefname{lemma}{Lemma}{Lemmas}
\crefname{theorem}{Theorem}{Theorems}
\crefname{assumption}{Assumption}{Assumptions}
\crefname{appendix}{Appendix}{Appendices}
\crefname{section}{Section}{Sections}
\crefname{example}{Example}{Examples}

\DeclareMathOperator{\He}{He}
\DeclareMathOperator{\co}{co}
\newcommand{\R}{\mathbb{R}}
\newcommand{\Sset}{\mathbb{S}}

\crefname{equation}{}{}

\definecolor{codeblue}{RGB}{22,74,132}
\definecolor{codegreen}{RGB}{20,110,62}
\lstdefinestyle{matlab}{
  language=Matlab,
  basicstyle=\ttfamily\footnotesize,
  keywordstyle=\color{codeblue}\bfseries,
  commentstyle=\color{codegreen},
  columns=fullflexible,
  keepspaces=true,
  frame=single,
  breaklines=true,
  showstringspaces=false,
  deletekeywords={grid,gamma}
}

\hypersetup{
    pdftitle={GriD-LMIA: The Gridding-Based Differentiable Parameter-Dependent LMI Assembler},
    pdfauthor={Yicheng Xu and Faryar Jabbari}
}

\begin{document}

\sptitle{Tools Paper}

\title{GriD-LMIA: The Gridding-Based Differentiable Parameter-Dependent LMI Assembler}

\author{Yicheng Xu\affilmark{1}}

\author{Faryar Jabbari\affilmark{1}}

\affil{Department of Mechanical and Aerospace Engineering, University of California Irvine, Irvine, CA 92697 USA (e-mail: \href{mailto:yichex7@uci.edu}{yichex7@uci.edu}; \href{mailto:fjabbari@uci.edu}{fjabbari@uci.edu})}

\corresp{CORRESPONDING AUTHOR: Yicheng Xu (e-mail: \href{mailto:yichex7@uci.edu}{yichex7@uci.edu})}

\markboth{GriD-LMIA: THE GRIDDING-BASED DIFFERENTIABLE PARAMETER-DEPENDENT LMI ASSEMBLER}{Y. XU AND F. JABBARI}

\begin{abstract}
    Parameter dependent linear matrix inequalities require conditions to hold over a continuous domain. When the scheduling parameters vary with time, derivatives of parameter dependent decision variables may also enter the conditions. Since semidefinite programming solvers require finitely many constraints, we introduce GriD-LMIA, the Gridding-based Differentiable parameter dependent LMI Assembler. Here, differentiable describes the derivative-dependent LMI, while the assembled continuous piecewise-polynomial decisions are locally Lipschitz and are interpreted through their Clarke generalized derivatives. The package partitions a hyper-rectangular domain and represents known data and decisions on each cell with tensor Bernstein polynomials. Six finite certificates are available through YALMIP: the direct Bernstein and P\'olya coefficient tests together with dense and sparse Putinar and FullBox Gram constructions. Each assembled constraint retains its physical-cell, Bernstein-coefficient, and rate-vertex indexing, so users can inspect the finite model before invoking a semidefinite programming solver. The examples examine the balance among grid density, decision degree, certificate strength, and semidefinite-program size.
\end{abstract}


\maketitle

\section{Introduction}
\label{sec:introduction}

Linear Matrix Inequalities (LMIs) play an important role in modern control theory, including stability analysis and controller synthesis for linear time-invariant systems \cite{Boyd1994,Scherer2000}. When the system dynamics vary with time, however, verification under finitely many fixed LMIs is no longer sufficient. One common model for such dynamics is the Linear Parameter-Varying (LPV) framework, in which the system matrices depend on scheduling parameters \cite{Sename2025}. The corresponding Parameter Dependent (PD) LMI must hold at every point of the continuous scheduling domain. Moreover, time-varying scheduling parameters introduce derivatives of a PD Lyapunov matrix \cite{Wu1996}, which leads to a differentiable PD-LMI (DPD-LMI) \cite{Scherer2006,Geromel2023}. Here, we introduce GriD-LMIA, the Gridding-based Differentiable PD-LMI Assembler, a MATLAB toolbox for assembling finitely many sufficient conditions for DPD-LMIs.

Existing LPV methods obtain finite tests through polytopic, linear-fractional, or gridded representations \cite{Hoffmann2015,Sename2025}. Polytopic methods assume affine or polytopic parameter dependence, which is convenient for finite analysis, but may overbound the original model and introduce conservatism, as they frame nonlinear dependence in a polytope. Linear-Fractional Transformation (LFT) models instead, isolate the scheduling parameters in a feedback block \cite{Packard1994}. Multiplier methods such as Integral Quadratic Constraints (IQCs) can then provide finite sufficient conditions \cite{Koroglu2007,Veenman2016}. However, the required LFT representation and the selected multiplier class may introduce additional conservatism. Gridding methods, on the other hand, impose fewer structural requirements on the scheduling parameter \cite{Gahinet1996,Xu2024a}, but size of LMIs grows exponentially with the number of gridded directions. Recent advances in modeling systems \cite{Lofberg2004} and Semi-Definite Programming (SDP) solvers \cite{MOSEKApS2019} have made such finite gridded formulations practical for a wider range of problems.

These different system representations have led to software packages with different modeling assumptions. ROLMIP \cite{Agulhari2019} parses robust polynomial LMIs using polytopic representations, whereas R-RoMulOC \cite{Chamanbaz2015} supports randomized analysis over convex uncertainty sets. MATLAB's Robust Control Toolbox addresses uncertain LFT models. LPVTools and LPVcore \cite{Hjartarson2015,denBoef2021} provide broader LPV modeling, analysis, and synthesis capabilities, including gridded LPV system models, and LPVTools also supports native LFT/IQC analysis. Finite grid-point tests in gridded LPV analysis, however, provide only necessary conditions for the continuous problem and rely on denser grids for approximation. By contrast, GriD-LMIA is a gridding-based numerical solution package for continuous-domain DPD-LMIs. In the numerical section, LPVTools therefore serves only as a baseline native-IQC reference.

We consider DPD-LMIs on a hyper-rectangular parameter domain, following early linear-spline methods \cite{Masubuchi1998Spline,Hamada2005}. Users specify a grid in each scheduling direction, and the tensor product of these grids partitions the domain into cells. On each normalized cell, tensor Bernstein polynomials \cite{Farouki2012} represent the known data and continuous piecewise-polynomial decision variables. Shared decision variables and matrix coefficients on the same boundary enforce value continuity without extra interface constraints. By collecting all adjacent limiting derivatives at a cell boundary, the Clarke generalized-derivative result in \cite{Clarke1990} extends the DPD-LMI across the full domain. GriD-LMIA then translates the remaining cellwise polynomial inequalities into finitely many LMIs through six certificate paths: Direct Bernstein, P\'olya, dense and sparse Putinar, and dense and sparse FullBox \cite{Harris2023,Polya1928,Putinar1993}. These paths export finite constraints to YALMIP \cite{Lofberg2004}, and the two sparse paths replace dense Gram matrices with banded Gram matrices over Bernstein basis labels. Compared to general purpose Sum-Of-Squares (SOS) packages, e.g., SOSTOOLS \cite{Papachristodoulou2021}, who supports broader classes of polynomial positivity programs. The SOS assemblers in GriD-LMIA operate only on polynomial matrices represented in the cellwise tensor Bernstein basis. The numerical studies examine decision degree, grid density, certificate choice, and closely aligned coefficient formulations with ROLMIP. The source code, installation instructions, and documentation are available at \url{https://github.com/TheBigoranger/GriD-LMIA}.

The remainder of the paper is organized as follows. \Cref{sec:problem,sec:lipschitz} develop the modeling stage from the problem setting to the cellwise Bernstein representation. \Cref{sec:certificates} presents the finite certificate stage, while \cref{sec:simulation,sec:conclusion} report the numerical evidence and summarize the appropriate use and limitations of GriD-LMIA.

\section{Problem Setting}
\label{sec:problem}
\subsection{Notation}
For a set $ A $, $ \mathscr{P}(A):=\{J\mid J\subseteq A\} $ denotes its power set. The binomial coefficient is $ \binom{m}{i}:=\frac{m!}{i!(m-i)!} $ for $ 0\leq i\leq m $, and zero otherwise. For block entries $ X[r,k] $, $ r\in R $ index the block rows and $ k\in K $ index the block columns. We write $ \left[X[r,k]\right]_{\substack{r\in R\\k\in K}} $ for the block matrix and $ \left[X[r]\right]_{r\in R} $ for a block column. The notation $ \operatorname{diag}\{X_q\} $ denotes a block-diagonal matrix. The transpose of $ X $ is $ X^\top $, $ \He(X)=X+X^\top $, $ \otimes $ denotes the Kronecker product, and $ I $ denotes an identity matrix. The set of real symmetric $ n\times n $ matrices is denoted by $ \Sset^n $. The Euclidean vector norm is $ \lVert\cdot\rVert_2 $. For a vector-valued signal $ f $, define $ \lVert f\rVert_{\mathcal L_2}:= \left(
    \int_0^\infty f(t)^\top f(t)\,dt
    \right)^{1/2} $, and let $ \mathcal L_2 $ denote the space of all signals with a finite $ \mathcal L_2 $ norm. For a symmetric matrix $ X $, the relations $ X\prec(\preceq)0 $, mean that $ X $ is negative (semi)definite respectively. Operator $ \co(\cdot) $ denotes the convex hull.

\subsection[LPV L2-gain analysis]{Motivation: LPV $ \mathcal L_2 $-gain analysis}
Consider an LPV system
\begin{gather}
    \label{eq: lpv-system}
    \Sigma:\left\{\begin{aligned}
        \dot x(t) & =A(\bm{\rho}(t))x(t) +B(\bm{\rho}(t))w(t), \\
        z (t)     & =C(\bm{\rho}(t))x(t) +D(\bm{\rho}(t))w(t).
    \end{aligned}\right.
\end{gather}
where $ x(t)\in\R^{n_x} $ is the state, $ w(t)\in\R^{n_w} $ is the exogenous input, and $ z(t)\in\R^{n_z} $ is the performance output. The matrices $ A $, $ B $, $ C $, and $ D $ have compatible dimensions, and $ \bm{\rho}(t)=\begin{bmatrix} \rho_1(t) & \ldots & \rho_\ell(t)
    \end{bmatrix}\in
    \mathcal P\subset\R^{\ell} $ is a measured scheduling parameter, with rate $ \dot{\bm\rho}(t)\in\mathcal R\subset\R^{\ell} $.
For zero initial conditions, the induced-$\mathcal L_2 $ gain of \cref{eq: lpv-system} is
\begin{gather}
    \label{eq: induced-l2-gain}
    \lVert
    \Sigma
    \rVert_{\mathcal L_2}=\sup_{w\in \mathcal{L}_2\setminus\{0\}}
    \frac{\lVert z\rVert_{\mathcal{L}_2}}{\lVert w\rVert_{\mathcal{L}_2}}.
\end{gather}
The following classical parameter-dependent condition motivates the derivative-dependent inequalities considered in this paper.
\begin{lemma}[Bound on $\mathcal{L}_2$ gain \cite{Scherer1996}]
    \label{lemma: L2-gain}
    Suppose there exist a differentiable matrix function $ P:\mathcal P\to\Sset^{n_x} $ and a scalar $ \gamma>0 $ such that $ P(\bm\rho)\succ0 $ and, for every $ (\bm\rho,\bm v)\in\mathcal P\times\mathcal R $,
    \begin{equation}
        L(\bm{\rho},\bm v) =\begin{bmatrix}
            \displaystyle\sum_{s=1}^\ell \frac{\partial P}{\partial\rho_s}v_s+\He(PA) & PB        & C^\top    \\
            B^\top P                                                                  & -\gamma I & D^\top    \\
            C                                                                         & D         & -\gamma I
        \end{bmatrix} \prec0,
        \label{eq: bounded-real}
    \end{equation}
    then $ \lVert \Sigma \rVert_{\mathcal L_2}\leq\gamma $.
\end{lemma}

Enforcing \cref{eq: bounded-real} over the continuous domain $ \mathcal P\times\mathcal R $ is a semi-infinite programming problem, whereas a SDP solver only accepts finitely many matrix constraints. GriD-LMIA therefore constructs sufficient conditions that can be exported to a numerical SDP model. The following assumption is adopted throughout the paper.
\begin{assumption}[hyper-rectangular parameter and rate]
    \label{assumption: hyper-rectangular parameter and rate}
    The domain $ \mathcal{P} $ of the scheduling parameter $ \bm{\rho}(t) $ and the domain $ \mathcal{R} $ of its rate $ \dot{\bm{\rho}}(t) $ are both hyper-rectangles:
    \begin{gather*}
        \mathcal{P}=\prod_{s=1}^{\ell}[\underline{\rho}_s,\overline{\rho}_s],\quad
        \mathcal{R}=\prod_{s=1}^{\ell}[\underline{\nu}_s,\overline{\nu}_s].
    \end{gather*}
\end{assumption}

Under \cref{assumption: hyper-rectangular parameter and rate}, the dependence of \cref{eq: bounded-real} on the rate $ \dot{\bm\rho} $ is affine. By convexity, it is equivalent to check \cref{eq: bounded-real} on $ (\bm{\rho},\dot{\bm{\rho}})\in \mathcal{P}\times\mathcal{V}_{\mathcal{R}} $, where $ \mathcal{V}_{\mathcal{R}} $ is the set of vertices of $ \mathcal{R} $:
\begin{gather*}
    \mathcal{V}_{\mathcal{R}}=\prod_{s=1}^{\ell}\{\underline{\nu}_s,\overline{\nu}_s\}.
\end{gather*}
To represent the strict inequality in \cref{eq: bounded-real} by a closed semidefinite constraint, let a residual margin be $ \epsilon_L>0 $. The rate dependence is then handled by the finite enumeration
\begin{gather}
    \operatorname{diag}\left\{L(\bm{\rho},\dot{\bm{\rho}})\mid
    \dot{\bm{\rho}}\in\mathcal{V}_{\mathcal{R}}\right\}\preceq-\epsilon_L I,
    \quad \forall \bm{\rho}\in\mathcal{P}.
    \label{eq: enumerated-lmi}
\end{gather}

However, the dependence on $ \bm{\rho} $ does not allow a simple convex combination to remove the continuum, which is the central modeling problem addressed by GriD-LMIA.

\subsection{Differentiable Parameter Dependent LMIs}

Let $ \bm{y}:\mathcal P\to\R^N $ be a differentiable parameter dependent decision vector $ \bm y(\bm\rho)=\begin{bmatrix}
        y_1(\bm\rho) & \ldots & y_N(\bm\rho)
    \end{bmatrix}^\top $. A general DPD-LMI considered in GriD-LMIA is
\begin{multline}
    \mathcal F(\bm{\rho})
    =  \sum_{i=1}^{N}\sum_{s=1}^{\ell}T_{i,s}(\bm{\rho})
    \frac{\partial y_i}{\partial\rho_s}(\bm\rho) \\
    +F_0(\bm{\rho})+\sum_{i=1}^{N}F_i(\bm{\rho})y_i(\bm{\rho})\preceq0,
    \quad \forall \bm{\rho}\in\mathcal{P}.
    \label{eq: dpdlmi-template}
\end{multline}
Here $ T_{i,s}(\bm\rho) $, $ F_0(\bm\rho) $, and $ F_i(\bm\rho) $ are known symmetric parameter dependent matrices of compatible dimensions. The \cref{eq: bounded-real} fits in \cref{eq: dpdlmi-template} by applying the margin shift in \cref{eq: enumerated-lmi}, absorbing $ \epsilon_L I $ into $ F_0 $, and vectorizing $ P(\bm\rho) $ together with the constant $ \gamma(\bm\rho)\equiv\gamma $ into $ \bm y(\bm\rho) $. \Cref{lem: clarke} extends this differentiable template to the locally Lipschitz decisions produced by gridding.

If $ T_{i,s}(\bm\rho)\equiv0 $ for every $ i,s $, the derivative terms disappear and \cref{eq: dpdlmi-template} reduces to a PD-LMI
\begin{equation}
    F_0(\bm\rho)+\sum_{i=1}^{N}F_i(\bm\rho)y_i(\bm\rho)\preceq0,
    \quad \forall \bm\rho\in\mathcal P.
    \label{eq:pdlmi-template}
\end{equation}

\section{Lipschitz Decision Variables and Bernstein Parameterization}
\label{sec:lipschitz}

The major difficulty in \cref{eq: dpdlmi-template} is that the condition must hold for every $ \bm\rho\in\mathcal P $. GriD-LMIA begins by partitioning $ \mathcal P $ into finitely many physical cells and retains the continuous dependence inside each cell through local Bernstein polynomials.

\subsection{Piecewise polynomial gridding}
\label{subsec: Piecewise polynomial gridding}
Since \cref{assumption: hyper-rectangular parameter and rate} holds, we can partition the parameter domain $ \mathcal P $ into finitely many cells, with the axis grid:
\begin{equation}
    \label{eq: axis-grid}
    \mathcal G_s=\{\underline\rho_s=\rho_s^{(1)}<\cdots<
    \rho_s^{(k_s)}=\overline\rho_s\},
    \quad s=1,\ldots,\ell.
\end{equation}
where $ k_s\geq2 $ is the number of grid nodes along scheduling direction $ s $. The tensor grid $ \mathcal G=\mathcal G_1\times\cdots\times\mathcal G_\ell $ contains $ \prod_{s=1}^\ell k_s $ physical nodes and partitions $ \mathcal P $ into $ \prod_{s=1}^\ell(k_s-1) $ physical cells. A hyper-rectangular cell is indexed by a vector $ \mathbf c=(c_1,\ldots,c_\ell) $, in which each entry $ c_s\in\{1,\ldots,k_s-1\} $:
\begin{equation}
    \label{eq: hyper-rectangle-cell}
    \mathcal H_{\mathbf c}=\prod_{s=1}^{\ell}
    [\rho_s^{(c_s)},\rho_s^{(c_s+1)}].
\end{equation}
A two-dimensional example is shown in \cref{fig: two-dimensional-grid}. It has $ 4\times3=12 $ physical nodes and $ 3\times2=6 $ physical cells. The cell $ \mathcal H_{(2,1)}= \left[\rho_1^{(2)},\rho_1^{(3)}\right] \times \left[\rho_2^{(1)},\rho_2^{(2)}\right] $ is highlighted in blue.

\begin{figure}[htbp]
    \centering
    \begin{tikzpicture}[x=1cm,y=1cm]
        \def\cellWidth{1.40}
        \def\cellHeight{1.1}
        \def\axisLead{0.15cm}
        \def\axisExtension{0.45cm}
        \def\gridLineWidth{0.55pt}
        \def\axisLineWidth{0.80pt}
        \def\labelGap{2pt}
        \def\nodeLabelGap{1pt}
        \def\gridCalloutGap{4pt}
        \def\gridPointerWidth{0.45pt}

        \tikzset{
            grid line/.style={draw=gray!70,line width=\gridLineWidth},
            grid point/.style={circle,fill,inner sep=1.7pt},
            cell label/.style={inner sep=0pt},
            axis label/.style={inner sep=0pt},
            axis grid label/.style={inner sep=0pt}
        }

        \coordinate (n11) at (0,0);
        \coordinate (n21) at (\cellWidth,0);
        \coordinate (n31) at ({2*\cellWidth},0);
        \coordinate (n41) at ({3*\cellWidth},0);
        \coordinate (n12) at (0,\cellHeight);
        \coordinate (n22) at (\cellWidth,\cellHeight);
        \coordinate (n32) at ({2*\cellWidth},\cellHeight);
        \coordinate (n42) at ({3*\cellWidth},\cellHeight);
        \coordinate (n13) at (0,{2*\cellHeight});
        \coordinate (n23) at (\cellWidth,{2*\cellHeight});
        \coordinate (n33) at ({2*\cellWidth},{2*\cellHeight});
        \coordinate (n43) at ({3*\cellWidth},{2*\cellHeight});

        \path (n11)--(n22) coordinate[pos=.5] (h11);
        \path (n21)--(n32) coordinate[pos=.5] (h21);
        \path (n31)--(n42) coordinate[pos=.5] (h31);
        \path (n12)--(n23) coordinate[pos=.5] (h12);
        \path (n22)--(n33) coordinate[pos=.5] (h22);
        \path (n32)--(n43) coordinate[pos=.5] (h32);

        \path let \p1=(n11), \p2=(n41), \p3=(n13) in
        coordinate (xAxisStart) at (\x1-\axisLead,\y1)
        coordinate (xAxisEnd) at (\x2+\axisExtension,\y2)
        coordinate (yAxisStart) at (\x1,\y1-\axisLead)
        coordinate (yAxisEnd) at (\x3,\y3+\axisExtension);

        \path[fill=blue!15] (n21) rectangle (n32);
        \draw[grid line] (n11)--(n41)--(n43)--(n13)--cycle;
        \foreach \from/\to in {n21/n23,n31/n33,n12/n42}{
                \draw[grid line] (\from)--(\to);
            }
        \foreach \point in {n11,n21,n31,n41,n12,n22,n32,n42,n13,n23,n33,n43}{
                \node[grid point] at (\point) {};
            }

        \draw[->,line width=\axisLineWidth] (xAxisStart)--(xAxisEnd);
        \draw[->,line width=\axisLineWidth] (yAxisStart)--(yAxisEnd);
        \node[axis label,right=\labelGap of xAxisEnd] {$ \rho_1 $};
        \node[axis label,above right=\labelGap of yAxisEnd] {$ \rho_2 $};

        \node[below=\nodeLabelGap of n11] (rho11) {$ \rho_1^{(1)} $};
        \node[below=\nodeLabelGap of n21] (rho12) {$ \rho_1^{(2)} $};
        \node[below=\nodeLabelGap of n31] (rho13) {$ \rho_1^{(3)} $};
        \node[below=\nodeLabelGap of n41] (rho14) {$ \rho_1^{(4)} $};
        \node[left=\nodeLabelGap of n11] (rho21) {$ \rho_2^{(1)} $};
        \node[left=\nodeLabelGap of n12] (rho22) {$ \rho_2^{(2)} $};
        \node[left=\nodeLabelGap of n13] (rho23) {$ \rho_2^{(3)} $};

        \node[axis grid label,right=\gridCalloutGap of rho14] (xGridLabel)
        {$ \rightarrow\mathcal G_1 $};
        \node[axis grid label,above=\gridCalloutGap of rho23, align=center] (yGridLabel)
        {$ \mathcal G_2 $\\
            $ \uparrow $};

        \node[cell label] at (h11) {$ \mathcal H_{(1,1)} $};
        \node[cell label,blue!70!black] at (h21) {$ \mathcal H_{(2,1)} $};
        \node[cell label] at (h31) {$ \mathcal H_{(3,1)} $};
        \node[cell label] at (h12) {$ \mathcal H_{(1,2)} $};
        \node[cell label] at (h22) {$ \mathcal H_{(2,2)} $};
        \node[cell label] at (h32) {$ \mathcal H_{(3,2)} $};
    \end{tikzpicture}
    \caption{Two-dimensional tensor gridding example on $ \mathcal{P}\subset \mathbb{R}^2 $.}
    \label{fig: two-dimensional-grid}
\end{figure}
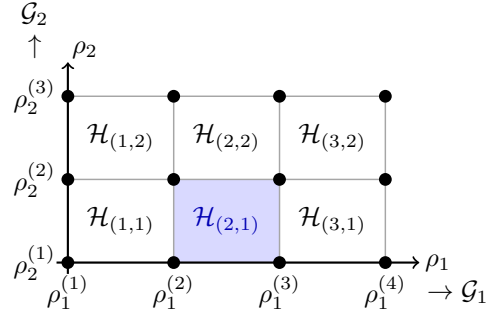

Any $ \bar{\bm\rho} $ that belongs to more than one cell is called an \textit{interface point}. For example, $ (\rho_1^{(2)},\rho_2^{(1)}) $ is shared by $ \mathcal H_{(1,1)} $ and $ \mathcal H_{(2,1)} $. We introduce the following assumption to ensure that the known data and decision variables are compatible with the gridding.

\begin{assumption}[Continuous piecewise polynomial data and decisions]
    \label{assumption: cell-polynomial}
    The known PD data $ T_{i,s}(\bm{\rho}),F_{0}(\bm{\rho}),F_{i}(\bm{\rho}) $ and the decision variables $ \bm{y}(\bm{\rho}) $ have the following properties:
    \begin{enumerate}
        \item They are polynomial with known degree on each hyper-rectangular cell $ \mathcal{H}_{\mathbf{c}}, \forall \mathbf{c}\in \prod_{s=1}^\ell \{1,\ldots,k_s-1\} $.
        \item Globally on $\mathcal{P}$, the known matrices are continuous, and the decisions $ \bm{y}(\bm{\rho}) $ are Lipschitz continuous.
    \end{enumerate}
\end{assumption}

The second statement equates polynomial values on a shared boundary of the cells, but does not require adjacent gradients to agree. Thus, the functions are continuous and piecewise $ C^1 $, but need not be globally $ C^1 $. Let $ f:\mathcal P\to\R^q $ be continuous on the finite cell partition, with each cell restriction extending to a $ C^1 $ map on a neighborhood of its closure. Let $ \mathcal A(\bar{\bm\rho}) $ collect the adjacent cell of an interface point $ \bar{\bm\rho} $. Its relative Clarke generalized Jacobian is
\begin{equation}
    \partial_{\mathcal P}^{\circ}f(\bar{\bm\rho})
    :=\co\left\{
    \lim_{\substack{
        \bm\rho\to\bar{\bm\rho}\\
        \bm\rho\in\operatorname{int}_{\mathcal P}(\mathcal H_{\mathbf c})
    }}
    Jf(\bm\rho)
    \;\middle|\;
    \mathbf c\in\mathcal A(\bar{\bm\rho})
    \right\},
    \label{eq:relative-clarke}
\end{equation}
where 
$ Jf(\bm\rho):=[\partial f_a/\partial\rho_s]_{\substack{
    a=1,\ldots,q\\
    \,s=1,\ldots,\ell
    }}\in\R^{q\times\ell} $ 
denotes the ordinary Jacobian of $ f $ at $ \bm\rho $ and $ \operatorname{int}_{\mathcal P} $ is relative to $ \mathcal P $. This definition gives the usual Clarke generalized Jacobian in the interior and the one-sided object on an exterior face \cite{Clarke1990}. 

\begin{lemma}[Interface coverage]
    \label{lem: clarke}
    Let $ \bm y $ be a decision vector satisfying \cref{assumption: cell-polynomial}. At an interface point $ \bar{\bm\rho} $, define the limiting Jacobian from $ \mathcal H_{\mathbf c} $ by
    \begin{equation}
        G_{\mathbf c}:=
        \left[
            \lim_{\substack{\bm\rho\to\bar{\bm\rho}\\
                    \bm\rho\in\operatorname{int}(\mathcal H_{\mathbf c})}}
            \frac{\partial y_i}{\partial\rho_s}(\bm\rho)
            \right]_{\substack{i=1,\ldots,N\\s=1,\ldots,\ell}}
        \in\R^{N\times\ell}.
        \label{eq:interface-jacobian}
    \end{equation}
    Write $ \partial_{\mathcal P}^{\circ}\bm y(\bar{\bm\rho}) $ for the relative Clarke generalized Jacobian in \eqref{eq:relative-clarke}. It is contained in the convex hull of the adjacent limiting Jacobians from cells inside $ \mathcal P $:
    \begin{equation}
        \partial_{\mathcal P}^{\circ} \bm y(\bar{\bm\rho})\subseteq
        \co\{G_{\mathbf c}:\mathbf c\in
        \mathcal A(\bar{\bm\rho})\}.
        \label{eq: clarke-hull}
    \end{equation}
    Suppose for all interface points $ \bar{\bm{\rho}} $:
    \begin{equation}
        \sum_{i=1}^{N}\sum_{s=1}^{\ell}
        T_{i,s}(\bar{\bm\rho})(G_{\mathbf c})_{i,s}
        +F_0(\bar{\bm\rho})
        +\sum_{i=1}^{N}F_i(\bar{\bm\rho})
        y_i(\bar{\bm\rho})
        \preceq0
        \label{eq: interface-cells}
    \end{equation}
    holds for every $ \mathbf c\in\mathcal A(\bar{\bm\rho}) $.  Then the same instance of \cref{eq: dpdlmi-template} satisfies
    \begin{equation}
        \sum_{i=1}^{N}\sum_{s=1}^{\ell}
        T_{i,s}(\bar{\bm\rho})G_{i,s}
        +F_0(\bar{\bm\rho})
        +\sum_{i=1}^{N}F_i(\bar{\bm\rho})
        y_i(\bar{\bm\rho})
        \preceq0
    \end{equation}
    for every $ G\in\partial_{\mathcal P}^{\circ}\bm y(\bar{\bm\rho}) $.
\end{lemma}

\begin{IEEEproof}
    By \cref{assumption: cell-polynomial}, $ \bm y $ is continuous and has $ C^1 $ cell extensions with bounded gradients on compact closures. Hence, $ \bm y $ is Lipschitz relative to $ \mathcal P $. Because the partition is finite, any sequence of differentiability points approaching an interface has a subsequence associated with one adjacent cell, with limiting Jacobian $ G_{\mathbf c} $. Taking the convex hull of all such limits gives \cref{eq: clarke-hull} \cite{Clarke1990}.

    Write $ G=\sum_{\mathbf c}\lambda_{\mathbf c}G_{\mathbf c} $ over adjacent cells, where $ \lambda_{\mathbf c}\geq0 $ and $ \sum_{\mathbf c}\lambda_{\mathbf c}=1 $. Affinity in $ G $ and common interface values express the DPD-LMI as the same convex combination of the cell-side expressions in \cref{eq: interface-cells}. Since the negative-semidefinite cone is convex, the combination remains negative semidefinite and proves the result.
\end{IEEEproof}

\begin{corollary}[Locally Lipschitz $\mathcal{L}_2$ gain]
    \label{cor:locally-lipschitz-bounded-real}
    Let $ \bm\rho $ be absolutely continuous with $ \bm\rho(t)\in\mathcal P $ and $ \dot{\bm\rho}(t)\in\mathcal R $ almost everywhere. Suppose $ P:\mathcal P\to\Sset^{n_x} $ is continuous, each cell restriction extends to a $ C^1 $ map on a neighborhood of its closure, and $ P(\bm\rho)\succ0 $. In every closed cell, replace the partial derivatives in \cref{eq: bounded-real} by the derivatives of that cell restriction. If, for some $ \epsilon_L>0 $, the resulting matrix satisfies $ L(\bm\rho,\bm v)\preceq-\epsilon_L I $ for every $ \bm\rho $ in the cell and every $ \bm v\in\mathcal V_{\mathcal R} $, then $ \lVert\Sigma\rVert_{\mathcal L_2}\leq\gamma $.
\end{corollary}
\begin{IEEEproof}
    The finite cell partition and compactness make $ P $ locally Lipschitz and give uniform positivity and a uniform residual margin. Hence $ P\circ\bm\rho $ and $ V=x^\top P(\bm\rho)x $ are absolutely continuous. Away from interfaces, the ordinary chain rule applies. For an absolutely continuous scalar function, its derivative is zero almost everywhere on each level set. Thus, at almost every time spent on a grid face, the scheduling velocity is tangent to that face. Continuity makes the tangential derivatives of adjacent cell restrictions agree, while \cref{lem: clarke} covers their Clarke convex hull. The cell-side inequalities therefore give \cref{eq: bounded-real} almost everywhere. Affinity in $ \dot{\bm\rho} $ extends the rate-vertex inequalities to every $ \dot{\bm\rho}\in\mathcal R $. The Schur complement then gives
    \begin{equation*}
        \dot V+\gamma^{-1}z^\top z-\gamma w^\top w<0
    \end{equation*}
    almost everywhere. Integration from the zero initial condition, followed by the limit over the time horizon, yields $ \lVert z\rVert_{\mathcal L_2}\leq\gamma\lVert w\rVert_{\mathcal L_2} $.
\end{IEEEproof}

\begin{remark}
    \Cref{lem: clarke} includes the differentiable case. If adjacent gradients agree, the classical and Clarke gradients coincide. Otherwise, enforcing \cref{eq: interface-cells} on every closed cell $ \mathcal{H}_{\mathbf{c}} $, for $ \mathbf c\in\prod_{s=1}^\ell\{1,\ldots,k_s-1\} $, covers each shared interface through the Clarke generalized derivative. \Cref{cor:locally-lipschitz-bounded-real} completes the corresponding almost-everywhere dissipation argument without requiring a smoothing step. For the rest of paper, we will only consider the behaviour of DPD-LMIs on the interiors of the cells, and the interface conditions is automatically enforced by \cref{lem: clarke}.
\end{remark}

GriD-LMIA therefore searches for continuous piecewise-polynomial decisions that are locally Lipschitz on the boundary. Smooth approximation provides an option towards global $C^1$ decision variables, when a classical differentiable decision is required. Continuous maps admit arbitrarily close smooth approximations \cite[Theorem~6.21]{Lee2012a}, locally Lipschitz functionals admit smooth graph approximations of their generalized gradients \cite[Theorem~3.7]{CwiszewskiAleksander2002}, and finite-dimensional Lipschitz mappings admit local $ C^1 $ Lipschitz approximations \cite[Proposition~5.1]{Dymond2026}. For the $\mathcal{L}_2$ gain case, a recovered decision preserves the DPD-LMI whenever it is uniformly smaller than $\epsilon_L$ for all $\bm\rho$.

\subsection{Normalization on single cell and limitation of Monomials}

On each hyper-rectangular cell $ \mathcal{H}_\mathbf{c}, \mathbf{c}\in \prod_{s=1}^\ell \{1,\ldots,k_s-1\} $, define the width $ h^{\mathbf c}_{s} $ and local coordinate $ \alpha_s $ along direction $ s=1,\ldots,\ell $:
\begin{subequations}
    \begin{align}
        h^{\mathbf c}_{s}
         & =\rho_s^{(c_s+1)}-\rho_s^{(c_s)},                         \\
        \alpha_s
         & =\frac{\rho_s-\rho_s^{(c_s)}}{h^{\mathbf c}_{s}}\in[0,1].
    \end{align}
    \label{eq: local-coordinate}
\end{subequations}
Collect the local coordinates in $ \bm\alpha=(\alpha_1,\ldots,\alpha_\ell)\in[0,1]^\ell $. Equivalently, the affine map $ \phi_{\mathbf c}:[0,1]^\ell\to\mathcal H_{\mathbf c} $ is defined componentwise by $ \phi_{\mathbf c,s}(\alpha_s)=\rho_s^{(c_s)}+h^{\mathbf c}_{s}\alpha_s $. For any cell-wise quantity $ X $, write $ X^{(\mathbf c)}=X\circ\phi_{\mathbf c} $ for its pullback to the unit box. In a cell interior and by the corresponding one-sided limit on its boundary, the chain rule gives
\begin{equation}
    \frac{\partial y_i}{\partial\rho_s}
    \bigl(\phi_{\mathbf c}(\bm{\alpha})\bigr)
    =\frac{1}{h^{\mathbf c}_{s}}
    \frac{\partial y_i^{(\mathbf c)}}{\partial\alpha_s}
    (\bm{\alpha}).
    \label{eq:cell-chain-rule}
\end{equation}
The restriction of \cref{eq: dpdlmi-template} to $ \mathcal H_{\mathbf c} $ is equivalent to the following unit-box DPD-LMI for every $ \bm{\alpha}\in[0,1]^\ell $:
\begin{equation}
    \begin{aligned}
        \mathcal F^{(\mathbf c)}(\bm{\alpha})
        ={} & \sum_{i=1}^{N}\sum_{s=1}^{\ell}
        \frac{1}{h^{\mathbf c}_{s}} T_{i,s}^{(\mathbf c)}(\bm{\alpha})
        \frac{\partial y_i^{(\mathbf c)}}{\partial\alpha_s}(\bm{\alpha}) \\
            & +F_0^{(\mathbf c)}(\bm{\alpha})
        +\sum_{i=1}^{N}
        F_i^{(\mathbf c)}(\bm{\alpha})
        y_i^{(\mathbf c)}(\bm{\alpha})\preceq0.
    \end{aligned}
    \label{eq: cell-wise-dpdlmi}
\end{equation}
The change of coordinates preserves the algebraic form of the original DPD-LMI template, as affine composition preserves polynomiality. Thus \cref{assumption: cell-polynomial} keeps every term in \cref{eq: cell-wise-dpdlmi} a polynomial in new coordinates $ \bm{\alpha}\in[0,1]^\ell $.

\begin{remark}
    The change of variables from the hyper-rectangle $ \bm{\rho}\in\mathcal{H}_{\mathbf c} $ to the unit box $ \bm\alpha\in[0,1]^\ell $ gives every physical cell the same local coordinate domain, even when nonuniform axis grids produce cells of different widths. The remaining task becomes to construct finitely many sufficient conditions for \cref{eq: cell-wise-dpdlmi}, and enumerate it over all cells.
\end{remark}

Given that \cref{assumption: cell-polynomial} holds in each cell, one needs to express both the known PD data and decision variables in a polynomial basis. An intuitive pick would be monomial:
\begin{gather}
    \label{eq: tensor-monomial-basis}
    \begin{Bmatrix}
        1 & \alpha_s & \ldots & \alpha^m_s
    \end{Bmatrix},\ s=1,\ldots,\ell.
\end{gather}
for some degree $ m\geq0 $. A polynomial on the $ \ell $-dimensional unit box $ [0,1]^{\ell} $ in the tensor monomial basis is
\begin{equation}
    p^{(\mathbf c)}(\bm{\alpha})
    =\sum_{\mathbf i\in\{0,\ldots,m\}^{\ell}}
    \bm\alpha^{\mathbf i} p^{(\mathbf c)}[\mathbf i],
    \quad
    \bm{\alpha}^{\mathbf i}
    =\prod_{r=1}^{\ell}\alpha_r^{i_r}.
    \label{eq: tensor-monomial-representation}
\end{equation}
where $ p^{(\mathbf c)}[\mathbf i] $, for $ \mathbf i \in \{0,\ldots,m\}^{\ell} $, are the $ (m+1)^{\ell} $ known matrix coefficients and do not depend on $ \bm\alpha $. However, this choice makes continuity across boundary of the cell an explicit constraint, as the following example shows.

\begin{example}[Equality constraint introduced by monomial basis]
    Let the one-dimensional axis grid be
    \[
        \mathcal G_1=\{\rho_1^{(1)}=0<\rho_1^{(2)}=1
        <\rho_1^{(3)}=2\}.
    \]
    The corresponding physical cells are $ \mathcal H_1=[0,1] $ and $ \mathcal H_2=[1,2] $.  Their forward local coordinates, as defined by \cref{eq: local-coordinate}, are $ \alpha=\rho_1 $ and $ \beta=\rho_1-1 $. Consider the piecewise polynomial function:
    \begin{equation}
        p(\rho_1)=
        \begin{cases}
            p_1(\alpha)=a_0+a_1\alpha+a_2\alpha^2,
             & \rho_1\in\mathcal H_1, \\
            p_2(\beta)=b_0+b_1\beta+b_2\beta^2,
             & \rho_1\in\mathcal H_2.
        \end{cases}
        \label{eq: shared-face-example}
    \end{equation}
    Continuity at the shared boundary $ \rho_1 = \rho_1^{(2)} $ requires
    \begin{equation}
        p_1(1)=p_2(0)
        \quad\Longleftrightarrow\quad
        a_0+a_1+a_2=b_0.
        \label{eq: monomial-interface-constraint}
    \end{equation}
    \cref{fig: grid-continuity} demonstrates the continuity requirement and the derivative jump at the shared interface taking $ (a_0,a_1,a_2)=(0,1/4,3/4) $ and $ (b_0,b_1,b_2)=(1,-5/4,1) $.  The pieces share $ p(1)=1 $, and their limiting derivatives are different as expected: $ p'_-(1)=7/4 $ and $ p'_+(1)=-5/4 $.

    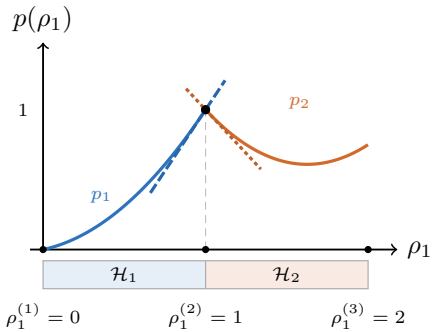
\begin{figure}[htbp]
        \centering
        \begin{tikzpicture}[x=2.15cm,y=1.85cm]
            \definecolor{leftcell}{RGB}{45,115,190}
            \definecolor{rightcell}{RGB}{210,105,45}

            \draw[->,line width=0.7pt] (-0.08,0)--(2.18,0)
            node[right] {$ \rho_1 $};
            \draw[->,line width=0.7pt] (0,0)--(0,1.48)
            node[above] {$ p(\rho_1) $};
            \draw[gray!55,densely dashed,line width=0.35pt]
            (1,0)--(1,1);
            \draw[leftcell,very thick,domain=0:1,samples=60]
            plot (\x,{0.25*\x+0.75*\x*\x});
            \draw[rightcell,very thick,domain=1:2,samples=60]
            plot (\x,{1-1.25*(\x-1)+(\x-1)*(\x-1)});
            \fill[black] (1,1) circle (1.7pt);
            \node[left=2pt,font=\scriptsize] at (0,1) {$ 1 $};

            \draw[leftcell!90!black,very thick,dash pattern=on 4pt off 1.4pt]
            (0.66,0.405)--(1.12,1.21);
            \draw[rightcell!90!black,very thick,dash pattern=on 1.5pt off 1.2pt]
            (0.88,1.15)--(1.34,0.575);
            \fill[black] (1,1) circle (1.7pt);
            \node[leftcell,font=\scriptsize,fill=white,inner sep=1pt]
            at (0.36,0.38) {$ p_1 $};
            \node[rightcell,font=\scriptsize,fill=white,inner sep=1pt]
            at (1.58,1.05) {$ p_2 $};

            \path[fill=leftcell!12] (0,-0.28) rectangle (1,-0.08);
            \path[fill=rightcell!13] (1,-0.28) rectangle (2,-0.08);
            \draw[gray!65,line width=0.45pt] (0,-0.28) rectangle (2,-0.08);
            \draw[gray!65,line width=0.45pt] (1,-0.28)--(1,-0.08);
            \node[font=\scriptsize] at (0.5,-0.18) {$ \mathcal H_1 $};
            \node[font=\scriptsize] at (1.5,-0.18) {$ \mathcal H_2 $};

            \foreach \x/\j/\value in {0/1/0,1/2/1,2/3/2} {
                    \fill[black] (\x,0) circle (1.35pt);
                    \node[below=2pt,font=\scriptsize] at (\x,-0.28)
                    {$ \rho_1^{(\j)}=\value $};
                }
        \end{tikzpicture}
        \caption{The grid $ \mathcal G_1 $ (lower strip) induces two cells.
            The continuous quadratic pieces have distinct limiting tangents at
            their shared grid point $ \rho_1^{(2)}=1 $.}
        \label{fig: grid-continuity}
    \end{figure}
\end{example}

As the parameter dimension $ \ell $ and the total node count $ \prod_{s=1}^{\ell} k_s $ increase, the number of interior codimension-one faces grows, at the speed of $ \sum_{s=1}^{\ell}(k_s-2)\prod_{t\neq s}(k_t-1) $. For tensor degree $ m $, each such face requires $ (m+1)^{\ell-1} $ coefficient equalities per scalar matrix entry in an explicit monomial construction. Moreover, assembling known matrices or decisions into \cref{eq: cell-wise-dpdlmi} would require updating the corresponding continuity constraints, which is tedious and error-prone. GriD-LMIA instead uses the endpoint property of the Bernstein basis \cite{Farouki2012}. A \texttt{pdvar} shares all $ (m+1)^{\ell-1} $ Bernstein coefficients on each common face, so the adjacent cell restrictions agree there. A \texttt{pdmat} constructed from continuous known data has the same agreement. Manually supplied cell pieces with inconsistent face values violate \cref{assumption: cell-polynomial} and are not repaired by the basis itself.

\subsection{Bernstein basis and coefficient algebra}

\label{sec: bernstein-basis}

On the one-dimensional unit interval $ [0,1] $, the Bernstein basis \cite{Farouki2012} of degree $ m $ is
\begin{equation}
    B_i^m(\alpha)=\binom{m}{i}(1-\alpha)^{m-i}\alpha^i,
    \qquad i=0,\ldots,m.
    \label{eq:bernstein-univariate}
\end{equation}

The polynomial expressed in the Bernstein basis is called a Bernstein polynomial, which has several useful properties.
\begin{itemize}
    \item The Bernstein polynomial is homogeneous.
    \item The Bernstein basis sums to one,
          $$ 1= \left[(1-\alpha)+\alpha\right]^d  =\displaystyle\sum_{i=0}^{d}B_{i}^{d}(\alpha).$$
    \item The derivative can be written as a combination of two Bernstein basis functions of lower degree:
          $$ \frac{d}{d\alpha}B_i^m(\alpha) = m\left[
                  B_{i-1}^{m-1}(\alpha)-B_{i}^{m-1}(\alpha)
                  \right].$$
    \item The Bernstein polynomial has the endpoint property:
          \begin{gather*}
              B_0^m(0)=B_m^m(1)=1, \\
              B_i^m(0)=0,\quad i=1,\ldots,m, \\
              B_i^m(1)=0,\quad i=0,\ldots,m-1.
          \end{gather*}
          Hence, if $ p(\alpha)=\sum_{i=0}^{m}c_iB_i^m(\alpha) $, then $ p(0)=c_0 $ and $ p(1)=c_m $.
\end{itemize}

The last property naturally admits the value continuity by reusing the same endpoint coefficients, thus requires no separate equality constraint. For example, the polynomial in \cref{eq: shared-face-example}, with the coefficients shown in \cref{fig: grid-continuity}, has the following equivalent degree-two Bernstein representation:
\begin{gather*}
    \hat{a}=[0,\tfrac18,1],
    \quad
    \hat{b}=[1,\tfrac38,\tfrac34], \\
    \left\{\begin{aligned}
        p_1(\alpha)=                                       & \hat{a}[1]B_0^2(\alpha)
        +\hat{a}[2]B_1^2(\alpha) +\hat{a}[3]B_2^2(\alpha), & \rho_1\in\mathcal H_1,  \\
        p_2(\beta)={}                                      & \hat{b}[1]B_0^2(\beta)
        +\hat{b}[2]B_1^2(\beta)+\hat{b}[3]B_2^2(\beta),    & \rho_1\in\mathcal H_2.
    \end{aligned}\right.
\end{gather*}
The endpoints are $ p(0)=0=\hat{a}[1] $ and $ p(2)=3/4=\hat{b}[3]$. At the interface point $ \rho_1^{(2)}=1 $, the shared Bernstein coefficient satisfies $ \hat{a}[3]=\hat{b}[1] $, which gives $ p_1(1)=p_2(0)=1 $. GriD-LMIA constructs this object from its axis grid and local Bernstein coefficients through \texttt{pdmat}:
\begin{lstlisting}[style=matlab]
A = pdmat({[0,1,2]}, {{0,1/8,1}, {1,3/8,3/4}});
\end{lstlisting}


Similarly, \texttt{pdvar} creates Bernstein decision variables by allocating shared YALMIP \cite{Lofberg2004} \texttt{sdpvar} handles for their local coefficients. The following declaration creates a scalar degree-two decision on the same axis grid:
\begin{lstlisting}[style=matlab]
P = pdvar(1,1, {[0,1,2]}, Degree=2);
\end{lstlisting}

For $ \ell>1 $, GriD-LMIA uses the tensor basis
\begin{equation}
    B_{\mathbf i}^m(\bm{\alpha})
    =\prod_{s=1}^{\ell}B_{i_s}^m(\alpha_s),
    \qquad \mathbf i\in\{0,\ldots,m\}^{\ell}.
    \label{eq:bernstein}
\end{equation}

On a fixed cell $ \mathcal H_{\mathbf c} $, a degree-$ m $ polynomial matrix has the representation
\begin{equation}
    P^{(\mathbf c)}(\bm\alpha)
    =\sum_{\mathbf i\in\{0,\ldots,m\}^{\ell}}
    B_{\mathbf i}^m(\bm{\alpha})
    P^{(\mathbf c)}[\mathbf i].
    \label{eq: tensor-bernstein-representation}
\end{equation}

Thus, within the cell $ \mathcal H_{\mathbf c} $, a tensor-product polynomial of degree at most $ m $ in each of the $ \ell $ parameter directions has $ (m+1)^\ell $ local coefficient matrices.  Under \cref{assumption: cell-polynomial}, the known matrices $ T_{i,s}^{(\mathbf c)}(\bm\alpha) $, $ F_0^{(\mathbf c)}(\bm\alpha) $, and $ F_i^{(\mathbf c)}(\bm\alpha) $ admit finite cellwise Bernstein representations, possibly of different tensor degrees, while the decision functions $ \bm y^{(\mathbf c)}(\bm\alpha) $ use the designated tensor degree $ m $.

Each of the following operations used to construct \cref{eq: cell-wise-dpdlmi} preserves linearity with respect to the independent Bernstein coefficients of the decision variables $ \bm{y}^{(\mathbf c)}[\mathbf{i}] $:
\begin{enumerate}
    \item Degree elevation from $ m $ to any $ M\geq m $ is a linear map and gives the unique degree-$ M $ representation of the same polynomial.
    \item Multiplication by a known degree-$ n $ polynomial matrix is linear in the decision coefficients and produces degree-$ m+n $ polynomial matrix.
    \item The partial derivative with respect to $ \alpha_s $ is linear in the decision coefficients and reduces the degree by one only along the $ s $-th direction.
    \item Addition and subtraction are linear after the operands have been elevated to compatible tensor degrees.
\end{enumerate}
\cref{app:bernstein-algebra} gives the exact coefficient identities and corresponding public-API examples.  After all terms have been elevated to a common tensor degree $ M $, GriD-LMIA represents the cellwise \cref{eq: cell-wise-dpdlmi} as
\begin{equation}
    \mathcal F^{(\mathbf c)}(\bm{\alpha})
    =\sum_{\mathbf i\in\{0,\ldots,M\}^{\ell}}
    B_{\mathbf i}^M(\bm{\alpha})
    \mathcal C^{(\mathbf c)}[\mathbf i]\preceq 0.
    \label{eq: cell-residual-bernstein}
\end{equation}
where
\begin{equation}
    \label{eq: cell-residual-coefficient}
    \mathcal C^{(\mathbf c)}[\mathbf i]
    =\mathcal C_0^{(\mathbf c)}[\mathbf i]
    +\sum_{i=1}^{N}
    \sum_{\mathbf j\in\{0,\ldots,m\}^{\ell}}
    \mathcal C_i^{(\mathbf c)}[\mathbf i,\mathbf j]\,
    y_i^{(\mathbf c)}[\mathbf j].
\end{equation}
Here $ \mathcal C_0^{(\mathbf c)}[\mathbf i] $ and $ \mathcal C_i^{(\mathbf c)}[\mathbf i,\mathbf j] $ are known matrices determined by the cell data, differentiation, coefficient convolution, and degree elevation. Thus, every coefficient is affine in the complete set of independent decision coefficients $ \bm y^{(\mathbf c)}[\mathbf j], \mathbf j \in\{0,\ldots,m\}^{\ell} $.

\begin{example}[$ m=1,\ell=1,N=1,k_1=2 $]
    Consider a normalized cell and suppress the cell superscript. Let the known matrix data $ A $, $ T $, and $ F $, as well as the scalar decision $ y $, all be linear splines on this cell. The derivative $ \dot{\alpha} $ is considered at two rate vertices: $ \nu\in\{-1,1\} $:
    \begin{equation}
        \begin{gathered}
            \mathcal F_{\nu}(\alpha)
            =A(\alpha)+\nu T(\alpha)
            \frac{\mathrm d y}{\mathrm d\alpha}(\alpha)
            +F(\alpha)y(\alpha), \\
            \begin{alignedat}{2}
                A(\alpha) & =\sum_{i=0}^1 B_i^1(\alpha)A[i],
                          & \quad T(\alpha)                  & =\sum_{i=0}^1 B_i^1(\alpha)T[i], \\
                F(\alpha) & =\sum_{i=0}^1 B_i^1(\alpha)F[i],
                          & \quad y(\alpha)                  & =\sum_{i=0}^1 B_i^1(\alpha)y[i].
            \end{alignedat}
        \end{gathered}
        \label{eq:univariate-residual-coefficient-example}
    \end{equation}
    Given the numeric matrices $ A[0] $, $ A[1] $, $ T[0] $, $ T[1] $, $ F[0] $, and $ F[1] $, the corresponding public construction is
    \begin{lstlisting}[style=matlab]
grid = {[0 1]};
rhorate = [-1 1];
A = pdmat(grid, {A0, A1}, Degree=1);
T = pdmat(grid, {T0, T1}, Degree=1);
F = pdmat(grid, {F0, F1}, Degree=1);
y = pdvar(1, grid, Degree=1);

dy = rhodiff(y, rhorate);
Fa = A + T * dy + F * y;
\end{lstlisting}
    Define $ \overline T=\tfrac12\bigl(T[0]+T[1]\bigr) $. Because $ F(\alpha)y(\alpha) $ has degree two, all terms are represented at the common degree $ M=2 $, which results in three Bernstein coefficients:
    \begin{gather*}
        \begin{bmatrix}
            \mathcal C^{(1)}_0[r]
        \end{bmatrix}_{r=0,1,2}
        =
        \begin{bmatrix}
            A[0]                          \\
            \tfrac12\bigl(A[0]+A[1]\bigr) \\
            A[1]
        \end{bmatrix},\\
        \begin{bmatrix}
            \mathcal C^{(1)}_{1}[r,k]
        \end{bmatrix}_{
            \substack{r=0,1,2\\k=0,1}}
        =
        \begin{bmatrix}
            F[0]-\nu T[0] & \nu T[0]                    \\
            \tfrac12F[1]-\nu\overline T
                          & \tfrac12F[0]+\nu\overline T \\
            -\nu T[1]     & F[1]+\nu T[1]
        \end{bmatrix}.
    \end{gather*}

    The three coefficients in $ \mathcal F(\bm{\alpha}) $ are:
    \begin{equation*}
        \begin{aligned}
            \mathcal C^{(1)}[0]
            ={} & A[0]+\bigl(F[0]-\nu T[0]\bigr)y[0]
            +\nu T[0]y[1],                                                                     \\
            \mathcal C^{(1)}[1]
            ={} & \tfrac12\bigl(A[0]+A[1]\bigr)   +\bigl(\tfrac12F[1]-\nu\overline T\bigr)y[0] \\
                & +\bigl(\tfrac12F[0]+\nu\overline T\bigr)y[1],                                \\
            \mathcal C^{(1)}[2]
            ={} & A[1]-\nu T[1]y[0]
            +\bigl(F[1]+\nu T[1]\bigr)y[1].
        \end{aligned}
    \end{equation*}
\end{example}

The operations above ultimately return a \texttt{pdvar} handle. GriD-LMIA collects the corresponding terms internally, so users need not enumerate the cells or manage the local coefficient enumeration within each cell.

\section{Finite Certificates in pdlmi}
\label{sec:certificates}
The preceding section represents DPD-LMI \cref{eq: cell-wise-dpdlmi} equivalently as the cell-wise Bernstein expansion \cref{eq: cell-residual-bernstein} under \cref{assumption: cell-polynomial}. The resulting matrix inequality remains semi-infinite because it must hold for every $ \bm{\alpha}\in[0,1]^\ell $.

Flipping the sign of $ \mathcal F^{(\mathbf c)}(\bm{\alpha}) $ gives
\begin{equation}
    -\mathcal F^{(\mathbf c)}(\bm{\alpha})=\sum_{\mathbf i\in\{0,\ldots,M\}^{\ell}}
    B_{\mathbf i}^{M}(\bm{\alpha})
    \bigl(-\mathcal C^{(\mathbf c)}[\mathbf i]\bigr)\succeq0,
    \label{eq:cell-positivity-certificate}
\end{equation}
In the scalar case, this condition reduces to nonnegativity of a real polynomial on a unit box. Algebraic positivity certificates provide finite sufficient representations of this condition \cite{Powers2021}. The term \textit{certificate} comes from real algebraic geometry and the Positivstellensatz. Here the coefficients $ \mathcal C^{(\mathbf c)}[\mathbf i] $ are matrix-valued, so the implemented conditions use matrix positivity certificates \cite{Savchuk2012,Scherer2006}. GriD-LMIA presents its six fixed-order certificates, which can be categorized into two families: coefficient-based and SOS based certificates.

\subsection{Coefficient-based certificate}
\label{sec: Coefficient-based certificate}

\paragraph{(Direct) Bernstein coefficient certificate}\mbox{}\par
The direct Bernstein sufficient test requires
\begin{equation}
    -\mathcal C^{(\mathbf c)}[\mathbf i]\succeq0
    \quad\forall \mathbf i \in\{0,\ldots,M\}^\ell.
    \label{eq: direct-certificate}
\end{equation}

The matrices $ \mathcal C^{(\mathbf c)}[\mathbf i] $ in \cref{eq: cell-residual-coefficient} are independent of $ \bm{\alpha} $ and affine in the decisions. Hence, \cref{eq: direct-certificate} is a finite set of LMIs that can be exported to YALMIP \cite{Lofberg2004}.
\begin{lstlisting}[style=matlab]
direct = Fa <= 0;
finiteConstraints = direct.toYalmip;
\end{lstlisting}

This test introduces one matrix condition per Bernstein coefficient and is the default used by the overloaded \texttt{<=} and \texttt{>=} operators. Grid refinement and degree elevation may reduce its conservatism, but neither of those guarantees feasibility at a prescribed finite order \cite{Farouki2012,Harris2023}.
\paragraph{P\'olya's relaxation}\mbox{}\par

Recall the second property of the Bernstein basis discussed in \cref{sec: bernstein-basis}. P\'olya's theorem implies that a strictly positive scalar polynomial on the box has positive coefficients after sufficiently large degree elevation.
\begin{lemma}[P\'olya's theorem in Bernstein coordinates \cite{Hardy1952,Polya1928}]
    Let $ p(\bm{\alpha}) $ have tensor Bernstein degree $ M $ in each of the $ \ell $ variables. For any $ d\geq0 $, and for all $ \bm{\alpha}\in[0,1]^\ell $:
    \begin{equation*}
        \begin{aligned}
            1
             & =\prod_{s=1}^{\ell}\left[(1-\alpha_s)+\alpha_s\right]^d  =
            \prod_{s=1}^{\ell}\left[
                                  \sum_{k_s=0}^{d}B_{k_s}^{d}(\alpha_s)\right] \\
             & =\sum_{\mathbf k\in\{0,\ldots,d\}^{\ell}}
            B_{\mathbf k}^{d}(\bm{\alpha})
        \end{aligned}
    \end{equation*}
    Hence, multiplying $ p $ by this identity gives its degree-$ M+d $ Bernstein representation
    \begin{equation*}
        \underbrace{\left[
                \sum_{\mathbf k\in\{0,\ldots,d\}^{\ell}}
                B_{\mathbf k}^{d}(\bm{\alpha})
                \right]}_{1}  p(\bm{\alpha})
        =\sum_{\mathbf j\in\{0,\ldots,M+d\}^{\ell}}
        B_{\mathbf j}^{M+d}(\bm{\alpha})
        \widehat{p}[\mathbf j].
    \end{equation*}
    If $ p(\bm{\alpha})>0 $ on $ [0,1]^\ell $, then its elevated Bernstein coefficients $ \widehat{p}[\mathbf j] $ are positive for a sufficiently large $ d $.
\end{lemma}

In practice, the actual $ d $ is unknown, but GriD-LMIA allows users to specify a nonnegative integer $ d $ to elevate the Bernstein degree of $ \mathcal F^{(\mathbf c)}(\bm{\alpha}) $ to $ M+d $, and impose the direct test on the elevated coefficients $ \widehat{\mathcal C}^{(\mathbf c)}[\mathbf j], \forall \mathbf j \in\{0,\ldots,M+d\}^\ell $.

\begin{lstlisting}[style=matlab]
direct = Fa <= 0;
polya = direct.usePolya(d);
finiteConstraints = polya.toYalmip;
\end{lstlisting}

\subsection{SOS-based Certificate}
\label{sec: SOS-based certificate}
Let $\bm{a}=\begin{bmatrix}
        r & \cdots & r
    \end{bmatrix}\in\mathbb{R}^\ell$, define the tensor Bernstein basis and related Gram form:
\begin{equation}
    \begin{aligned}
        \bm{b}_{\bm{a}}(\bm{\alpha})
         & =\bigotimes_{s=1}^{\ell}
            [B_0^{a_s}(\alpha_s),\ldots,B_{a_s}^{a_s}(\alpha_s)]^\top, \\
        S(\bm{\alpha})
         & =(\bm{b}_{\bm{a}}(\bm{\alpha})\otimes I)^\top Q
        (\bm{b}_{\bm{a}}(\bm{\alpha})\otimes I),
    \end{aligned}
    \label{eq:bernstein-gram-basis}
\end{equation}
with $Q\succeq0$. For a even $M$, $-\mathcal F^{(\mathbf c)}$ can be equivalently expressed into a Gram form with block-Hankel matrix. However, requiring this block-Hankel matrix to be positive semi-definite is generally restrictive, which is demonstrated by the following counterexample.
\begin{example}
    Consider one parameter, one physical cell, and residual degree $M=4$. Let the scalar Bernstein coefficients be
    \begin{gather*}
        \left[\mathcal C^{(1)}[i]\right]_{i=0}^{4}
        =\begin{bmatrix}
            -1 & -3 & -5 & -7 & -1
        \end{bmatrix}^\top.
    \end{gather*}
    Every coefficient is negative, so the Direct certificate in \cref{eq: direct-certificate} holds. However, the corresponding coefficient-induced block-Hankel matrix satisfies
    \begin{gather*}
        -\mathcal H^{(1)}
        =\begin{bmatrix}
            1 & 3 & 5 \\
            3 & 5 & 7 \\
            5 & 7 & 1
        \end{bmatrix}\not\succeq0,
    \end{gather*}
    because its leading $2\times2$ principal minor has determinant $-4$. Thus, the block-Hankel sign test fails even though the Direct certificate holds.
\end{example}
The counterexample presents a problem: the polynomial expressed by a single Gram matrix has its own limitation and forcing it to be negative semidefinite is too restrictive. An SOS certificate therefore, does not fix the Gram matrix from the coefficients. Instead, it introduces free positive-semidefinite Gram matrices in prescribed basis and matches the resulting polynomial identity coefficientwise in the Bernstein basis. This separates the polynomial identity from the choice of Gram matrix and allows the geometry of the cell to enter through nonnegative weights.

Note that on a normalized physical cell, these weights are supplied by the unit box $[0,1]^\ell$, which is a basic closed semi-algebraic set generated by
\begin{gather*}
    g_s(\bm{\alpha})=\alpha_s(1-\alpha_s)\geq 0, s=1,\ldots,\ell.
\end{gather*}
These specific box generators produce an Archimedean
quadratic module \cite{Powers2021,Shang2025}, which motivates denominator-free SOS identities. At absolute Gram order $r$, the basis degrees for the corresponding Gram forms are chosen as
\begin{equation}
    \begin{aligned}
        S_0(\bm{\alpha}):
         & \quad \bm a=(r,\ldots,r),                        \\
        g_s(\bm{\alpha})S_s(\bm{\alpha}):
         & \quad \bm a=(r,\ldots,
        \underbrace{r-1}_{s\text{-th direction}},\ldots,r), \\
        \left(\prod_{s\in J}g_s(\bm{\alpha})\right)S_J(\bm{\alpha}):
         & \quad a_t=\begin{cases}
                         r-1, & t\in J,    \\
                         r,   & t\notin J.
                     \end{cases}
    \end{aligned}
    \label{eq:gram-basis-degree-choice}
\end{equation}

\begin{example}
    \label{ex: gram-basis-degree-choice}
    When $ \ell=3 $, all Gram terms and their tensor Bernstein bases are shown in \cref{tab: weighted-gram-bases}.
    \begin{table}[htbp]
        \centering
        \caption{Weighted Gram terms and tensor Bernstein bases for $ \ell=3 $.}
        \label{tab: weighted-gram-bases}
        {\footnotesize
            \setlength{\tabcolsep}{3pt}
            \begin{tabular}{@{}cccc@{}}
                \toprule
                \textnormal{Term} & \textnormal{Basis}
                                  & \textnormal{Term}          & \textnormal{Basis}      \\
                \midrule
                $ S_0 $           & $ \bm b_{r,r,r} $
                                  & $ g_1g_2S_{\{1,2\}} $      & $ \bm b_{r-1,r-1,r} $   \\
                $ g_1S_1 $        & $ \bm b_{r-1,r,r} $
                                  & $ g_1g_3S_{\{1,3\}} $      & $ \bm b_{r-1,r,r-1} $   \\
                $ g_2S_2 $        & $ \bm b_{r,r-1,r} $
                                  & $ g_2g_3S_{\{2,3\}} $      & $ \bm b_{r,r-1,r-1} $   \\
                $ g_3S_3 $        & $ \bm b_{r,r,r-1} $
                                  & $ g_1g_2g_3S_{\{1,2,3\}} $ & $ \bm b_{r-1,r-1,r-1} $ \\
                \bottomrule
            \end{tabular}
        }
    \end{table}
\end{example}

An SOS-based certificate represents the polynomial matrix $ -\mathcal F^{(\mathbf c)} $ as a weighted sum of Gram forms and matches the resulting polynomial identity coefficient-wise in the Bernstein basis. For $ \ell=1 $, the scalar Markov--Luk\'acs theorem gives an exact characterization of scalar polynomial non-negativity on the unit interval using two Gram forms \cite{Roh2006,Powers2021}.
\begin{lemma}[Markov--Luk\'acs theorem]
    \label{lem: markov-lukacs}
    Let $ p(\alpha) $ be a real univariate polynomial of degree $ M $, and let $ r=\lfloor M/2\rfloor $. Then $ p(\alpha)\geq0 $ for every $ \alpha\in[0,1] $ if and only if it admits the representation
    \begin{equation}
        \label{eq: markov-lukacs}
        p(\alpha)=
        \begin{cases}
            s_0(\alpha)+\alpha(1-\alpha)s_1(\alpha),  & M=2r,   \\
            (1-\alpha)s_L(\alpha)+\alpha s_U(\alpha), & M=2r+1,
        \end{cases}
    \end{equation}
    where $ s_0 $, $ s_1 $, $ s_L $, and $ s_U $ are SOS polynomials. In the even case, their degrees are at most $ 2r $ and $ 2r-2 $, respectively, and the $ s_1 $ term is absent when $ r=0 $. In the odd case, both $ s_L $ and $ s_U $ have degree at most $ 2r $.
\end{lemma}

For a symmetric polynomial matrix inequality, GriD-LMIA applies the same weight matching to matrix-SOS Gram forms. This paper uses the resulting matrix Markov--Luk\'acs certificate as a sufficient condition and does not infer matrix exactness from the scalar theorem. In particular, an odd assembled degree $ M=2r+1 $ uses two complete degree-$ r $ Gram bases, weighted by $ 1-\alpha $ and $ \alpha $, respectively.

For $ \ell\geq2 $, GriD-LMIA provides four optional sufficient SOS paths: Putinar \cite{Putinar1993,Magron2015}, SparsePutinar, FullBox \cite{Magron2015}, and Sparse FullBox. All four assemblers use the even matching degree $ 2r $. An odd degree $ M $ is therefore elevated to degree $ M+1 $ before coefficient matching. For a user-specified $ r $, the assemblers elevate the residual $ \mathcal{F}^{(\mathbf c)} $ to degree $ 2r $ and then match its Bernstein coefficients.

The default is the minimum admissible absolute Gram order:
\begin{equation}
    r_{\mathrm{default}}=r_{\min}=
    \begin{cases}
        \lfloor M/2\rfloor, & \ell=1,    \\
        \lceil M/2\rceil,   & \ell\geq2,
    \end{cases}
    \label{eq:default-gram-order}
\end{equation}

\paragraph{Putinar quadratic-module certificate}\mbox{}\par
For $ \ell\geq2 $, the Archimedean quadratic-module structure of the box generators motivates the implemented Putinar certificate \cite{Putinar1993,Magron2015}:
\begin{equation}
    -\mathcal F^{(\mathbf c)}
    =S_0+\sum_{s=1}^{\ell}g_sS_s .
    \label{eq:putinar}
\end{equation}
\begin{lstlisting}[style=matlab]
direct = Fa <= 0;
putinar = direct.usePutinar(r);
finiteConstraints = putinar.toYalmip;
\end{lstlisting}

\paragraph{SparsePutinar banded-Gram certificate}\mbox{}\par
SparsePutinar retains the unweighted and singleton-weighted terms in \cref{eq:putinar} but replaces each dense Gram matrix with a banded Gram matrix. For a tensor basis indexed by $ \mathcal I_{\bm a}:=\prod_{t=1}^{\ell}\{0,\ldots,a_t\} $, write its block Gram matrix as $ Q=[Q[\mathbf i,\mathbf j]]_{\mathbf i,\mathbf j\in\mathcal I_{\bm a}} $. We say that $ Q $ has bandwidth $ \omega\geq1 $ when
\begin{gather*}
    Q[\mathbf i,\mathbf j]=0
    \quad\text{if}\quad
    \lVert\mathbf i-\mathbf j\rVert_\infty>\omega-1.
\end{gather*}
Written as $Q^{(\omega)}$ with bandwidth $ \omega $.
Let $ \bm e_s $ denote an indicator vector where the $ s $th entry is $1$ and $0$ otherwise. Define $ \bm a_0=r\bm 1 $ and $ \bm a_s=\bm a_0-\bm e_s $. The SparsePutinar certificate is
\begin{equation}
    \begin{aligned}
        -\mathcal F^{(\mathbf c)}(\bm\alpha)
                       & =S_0^{(\omega)}(\bm\alpha)
        +\sum_{s=1}^{\ell}g_s(\bm\alpha)S_s^{(\omega)}(\bm\alpha),   \\
        S_0^{(\omega)}
                       & =(\bm b_{\bm a_0}(\bm\alpha)\otimes I)^\top
        Q_0^{(\omega)}
        (\bm b_{\bm a_0}(\bm\alpha)\otimes I),                       \\
        S_s^{(\omega)}
                       & =(\bm b_{\bm a_s}(\bm\alpha)\otimes I)^\top
        Q_s^{(\omega)}
        (\bm b_{\bm a_s}(\bm\alpha)\otimes I),                       \\
        Q_0^{(\omega)} & \succeq0,\qquad Q_s^{(\omega)}\succeq0 .
    \end{aligned}
    \label{eq:sparse-putinar}
\end{equation}
\begin{lstlisting}[style=matlab]
omega = 2;
direct = Fa <= 0;
sparsePutinar = direct.useSpPut(omega, r);
finiteConstraints = sparsePutinar.toYalmip;
\end{lstlisting}
Here, $\omega$ is the common mathematical symbol for the sparse-control argument preceding $r$. The released \texttt{useSpPut} interface names and stores this argument as \texttt{cliqueSize}/\texttt{CliqueSize} which corresponds to the bandwidth above. The released \texttt{useSpBox} interface names and stores the corresponding argument as \texttt{bandWidth}/\texttt{BandWidth}.

\paragraph{FullBox preordering certificate}\mbox{}\par
For $ \ell\geq2 $, FullBox enlarges the Putinar index family to the complete power set $ \mathscr{P}(\{1,\ldots,\ell\}) $:
\begin{equation}
    -\mathcal F^{(\mathbf c)}
    =\sum_{J\in \mathscr{P}(\{1,\ldots,\ell\})}
    \left(\prod_{s\in J}g_s\right)S_J,
    \label{eq:preorder}
\end{equation}
where $ S_{\varnothing}=S_0 $. For example, using the case in \cref{ex: gram-basis-degree-choice}, Putinar certificate considers only the left column of \cref{tab: weighted-gram-bases}, while FullBox considers both columns.

\cref{eq:preorder} is the full preordering generated by the quadratic box polynomials $ g_s=\alpha_s(1-\alpha_s) $. Some authors refer to this as a restricted Schm\"udgen-type preordering on quadratic box generators \cite{Magron2015}.
\begin{lstlisting}[style=matlab]
direct = Fa <= 0;
fullBox = direct.useFullBox(r);
finiteConstraints = fullBox.toYalmip;
\end{lstlisting}

\paragraph{Sparse FullBox banded-Gram certificate}\mbox{}\par

Putinar and FullBox assembly introduce high-dimensional symmetric Gram matrices. For example, $ S_0 $ contains a Gram matrix $ Q_0 $ with $ (r+1)^\ell $ block rows and columns, each block having the dimension of $ \mathcal{C}^{(\mathbf c)}[\mathbf i] $. For $ \ell>1 $, the resulting Gram variables can dominate the size of the assembled SDP.

To exploit sparsity, consider the FullBox representation at a fixed even Bernstein degree $ M=2r $. The Direct certificate in \cref{eq: direct-certificate} is represented by diagonal Gram blocks associated with the unweighted basis coefficients.

Sparse FullBox relaxes this block-diagonal structure with banded Gram matrices while retaining the complete FullBox traversal. Let $ \mathbf 1_{\{s\in J\}} $ be an $\ell$-dimensional vector whose $s$th entry is one when $ s\in J $ and zero otherwise. Define $ \bm a_J=(a_{J,1},\ldots,a_{J,\ell}) $, where $ a_{J,s}=r-\mathbf 1_{\{s\in J\}} $ as prescribed by \cref{eq:gram-basis-degree-choice}. Its certificate is
\begin{equation}
    \begin{aligned}
        -\mathcal F^{(\mathbf c)}(\bm\alpha)
         & =\sum_{J\in\mathscr P(\{1,\ldots,\ell\})}
        \prod_{s\in J}g_s(\bm\alpha)\,S_J^{(\omega)}(\bm\alpha), \\
        S_J^{(\omega)}
         & =(\bm b_{\bm a_J}(\bm\alpha)\otimes I)^\top
        Q_J^{(\omega)}
        (\bm b_{\bm a_J}(\bm\alpha)\otimes I),                   \\
        Q_J^{(\omega)}
         & \succeq0 .
    \end{aligned}
    \label{eq:sparse-preorder}
\end{equation}
Each $ Q_J^{(\omega)} $ has bandwidth $ \omega $ under the preceding block definition. At $ \omega=2 $, adjacent basis labels are coupled, which gives a block-tridiagonal pattern in one parameter. The default is $ \omega=2 $.
\begin{lstlisting}[style=matlab]
omega = 2;
direct = Fa <= 0;
sparseFullBox = direct.useSpBox(omega, r);
finiteConstraints = sparseFullBox.toYalmip;
\end{lstlisting}

\begin{remark}
    When $ \ell=1 $, dense Putinar and FullBox use the matrix Markov--Luk\'acs-form representation induced by \cref{eq: markov-lukacs}, while the sparse paths replace the corresponding dense Gram blocks with banded Gram matrices. Increasing $ \omega $ admits wider bands, and a bandwidth spanning every active basis direction recovers the corresponding dense Gram matrix. Sparse FullBox at $ \omega=1 $ is implemented as the Direct certificate. The fixed-order certificate containment relations are summarized in \cref{fig:certificate-cone-burden}. Moving toward a smaller cone can reduce the maximum Gram-block dimension while imposing a stronger sufficient condition.

    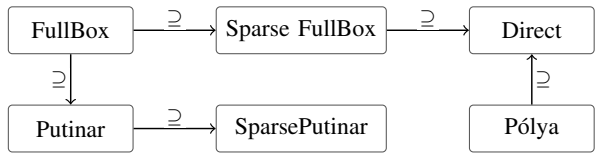
\begin{figure}[htbp]
    \centering
    \begin{tikzpicture}[
            certificate/.style={
                    draw=black!65,
                    rounded corners=1.5pt,
                    minimum height=0.62cm,
                    minimum width=1.65cm,
                    align=center,
                    font=\footnotesize
                },
            relation/.style={->,line width=0.55pt},
            relation label/.style={
                    fill=white,
                    inner sep=1pt,
                    font=\scriptsize,
                    sloped
                }
        ]
        \node[certificate] (fullbox) at (0,0) {FullBox};
        \node[certificate,minimum width=2.25cm]
        (sparse) at (3.05,0) {Sparse FullBox};
        \node[certificate] (direct) at (6.10,0) {Direct};
        \node[certificate] (putinar) at (0,-1.30) {Putinar};
        \node[certificate,minimum width=2.25cm]
        (sparseputinar) at (3.05,-1.30) {SparsePutinar};
        \node[certificate] (polya) at (6.10,-1.30) {P\'olya};

        \draw[relation] (fullbox) -- node[relation label,above] {$ \supseteq $}
        (sparse);
        \draw[relation] (sparse) -- node[relation label,above] {$ \supseteq $}
        (direct);
        \draw[relation] (fullbox) -- node[midway,left,fill=white,inner sep=1pt,font=\scriptsize] {$ \supseteq $}
        (putinar);
        \draw[relation] (putinar) -- node[relation label,above] {$ \supseteq $}
        (sparseputinar);
        \draw[relation] (polya) -- node[midway,right,fill=white,inner sep=1pt,font=\scriptsize] {$ \supseteq $}
        (direct);
    \end{tikzpicture}
    \caption{Implemented fixed-order certificate containments. Only the displayed inclusions are asserted.}
    \label{fig:certificate-cone-burden}
\end{figure}

    All six certificate paths impose sufficient conditions for the continuous semidefinite PD-LMI. For the strict $\mathcal{L}_2$ gain, the margin shift in \cref{eq: enumerated-lmi} is applied before certificate assembly. Its zero-margin closure is a finite semidefinite model but does not, by solver status alone, certify the strict $\mathcal{L}_2$ gain. Failure at a fixed $ m $, $ d $, $ r $, or $ \omega $ is inconclusive. The direct and P\'olya paths use coefficient inequalities, while the SOS paths add Gram variables and coefficient-matching identities through \cref{eq:putinar,eq:sparse-putinar,eq:sparse-preorder,eq:preorder}. The reported sizes and solver objectives describe the stated numerical configurations rather than a universal ordering.
\end{remark}

\section{Numerical Evaluation and Software Comparison}
\label{sec:simulation}

All experiments were run on 64-bit Windows 11 with an Intel Core i7-11700K processor and 32 GB of memory. The software environment comprised GriD-LMIA v1.4.0, MATLAB R2026a Update 3, YALMIP dated 26 June 2025, ROLMIP, LPVTools 2.0.0, SeDuMi 1.3.7, SDPT3 4.0, and MOSEK 11.2.

\subsection[One-parameter degree, grid, and certificate sensitivity]{Degree, grid, and certificate sensitivity when $ \ell=1 $}

The one-parameter study restores the example from \cite[Example 1]{Masubuchi1998Spline}. Consider the following LPV system:
\begin{gather}
    \begin{aligned}
        A(\rho) & =\begin{bmatrix}-1&0.5\\-1&-2\end{bmatrix}
        +\rho\begin{bmatrix}-1.3&-20\\2&-10\end{bmatrix},
        \label{eq:masubuchi-plant}                           \\
        B(\rho) & =\begin{bmatrix}1&-4\\-1&-1\end{bmatrix}
        +\rho\begin{bmatrix}2.2&0.5\\-6&-5\end{bmatrix}.
        \nonumber
    \end{aligned},\\
    C(\rho)=I_2,\quad D(\rho)=0.
    \label{eq: masubuchi-plant}
\end{gather}
with $ \rho\in[0,1] $ and $ \dot\rho\in[-1,1] $. The reported quantity is the smallest solver-returned $ \gamma $ for the assembled model, using the zero-margin DPD-LMI closure of \cref{eq: bounded-real} together with $ P(\rho)\succ0 $. The finite models use YALMIP \cite{Lofberg2004}, and all DPD-LMI configurations in this study were solved with SDPT3 \cite{Toh1999}.

The following code constructs the GriD-LMIA model for a prescribed grid-node count \texttt{k} and decision degree \texttt{deg}.
\begin{lstlisting}[style=matlab]
k=2; deg =1;
grid = linspace(0, 1, k);
\end{lstlisting}
\begin{lstlisting}[style=matlab]
A = pdmat(grid, @(rho) [-1, 0.5; -1, -2] ...
            + rho * [-1.3, -20; 2, -10], Degree=1);
B = pdmat(grid, @(rho) [1, -4; -1, -1] ...
            + rho * [2.2, 0.5; -6, -5], Degree=1);
C = eye(2);
D = zeros(2);
P = pdvar(2, grid, Degree = deg);
gamma = sdpvar(1);
diffP = rhodiff(P, [-1 1]);
DPDLMI = [diffP + P * A + A' * P, P * B, C';
    B' * P, -gamma * eye(2), D';
    C, D, -gamma * eye(2)] <= 0;
% Zero-margin closure used in the reported benchmark.
% DPDLMI = DPDLMI.usePolya(2);
posDefP = P >= 1e-8 * eye(2);
solver = 'sdpt3';
opts = sdpsettings('solver', solver, 'verbose', 0);
constraints = [DPDLMI.toYalmip, posDefP.toYalmip];
sol = optimize(constraints, gamma, opts);
assert(sol.problem == 0, sol.info);
value(gamma)
\end{lstlisting}

The study varies three modeling choices that affect the finite representation:
\begin{enumerate}
    \item The grid-node count $ k $, which creates $ k-1 $ intervals.
    \item The Bernstein degree $ m $ of the Lyapunov matrix $ P $ on each interval.
    \item The finite certificate: direct Bernstein coefficients, P\'olya elevation, or an SOS-based certificate.
\end{enumerate}

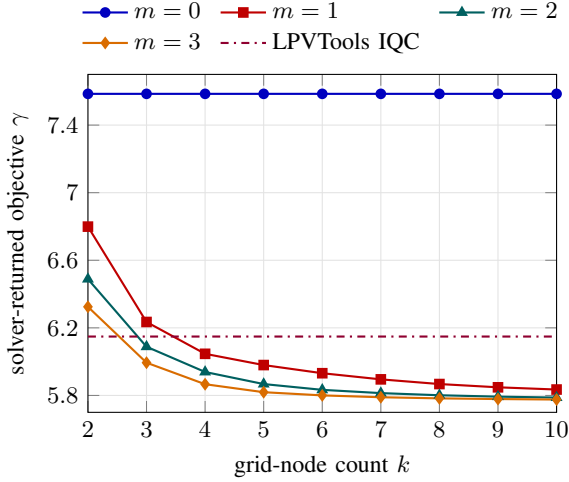
\begin{figure}[htbp]
    \centering
    \begin{tikzpicture}[trim axis left, trim axis right]
        \begin{axis}[
                masubuchi axis,
                xmin=2, xmax=10,
                xtick={2,3,...,10},
                xlabel={grid-node count $ k $},
                legend columns=3,
                legend cell align={left},
                legend style={
                        at={(0.5,1.03)},
                        anchor=south,
                        draw=none,
                        font=\footnotesize,
                        /tikz/every even column/.append style={column sep=5pt}
                    }]
            \addplot[blue!75!black,thick,mark=*,mark size=1.6pt]
            table[x=k,y=m0,col sep=comma]
                {pics/masubuchi_grid_sensitivity.csv};
            \addlegendentry{$ m=0 $}
            \addplot[red!75!black,thick,mark=square*,mark size=1.6pt]
            table[x=k,y=m1,col sep=comma]
                {pics/masubuchi_grid_sensitivity.csv};
            \addlegendentry{$ m=1 $}
            \addplot[teal!75!black,thick,mark=triangle*,mark size=1.8pt]
            table[x=k,y=m2,col sep=comma]
                {pics/masubuchi_grid_sensitivity.csv};
            \addlegendentry{$ m=2 $}
            \addplot[orange!85!black,thick,mark=diamond*,mark size=1.8pt]
            table[x=k,y=m3,col sep=comma]
                {pics/masubuchi_grid_sensitivity.csv};
            \addlegendentry{$ m=3 $}
            \addplot[purple!75!black,thick,dash dot,no marks]
            table[x=k,y=lpvtools,col sep=comma]
                {pics/masubuchi_grid_sensitivity.csv};
            \addlegendentry{LPVTools IQC}
        \end{axis}
    \end{tikzpicture}
    \caption{Degree--grid sensitivity under the direct Bernstein coefficient
        certificate, with the grid-free LPVTools rate-bounded IQC result as a
        baseline reference \cite{Hjartarson2015}.}
    \label{fig:masubuchi-grid}
\end{figure}

To show the effects of grid-node count and decision degree, we first fix the certificate to the direct Bernstein coefficient certificate in \cref{eq: direct-certificate}. \Cref{fig:masubuchi-grid} varies the Bernstein degree of the Lyapunov matrix, $ m\in\{0,1,2,3\} $, and the grid-node count, $ k\in\{2,\ldots,10\} $. The number of shared Bernstein coefficient matrices in $ P(\rho) $ is $ 1+m(k-1) $. Since $ P $ is symmetric and two dimensional, these matrices contain $ 3[1+m(k-1)] $ independent scalar decision entries.

At $ m=0 $, $ P(\rho) $ is constant. Refining the grid therefore adds no decision freedom, and the solver objective remains $ 7.58491 $, matching the reported constant-matrix result \cite{Masubuchi1998Spline}. For $ m=1,2,3 $, $ P(\rho) $ is a continuous piecewise-linear, piecewise-quadratic, or piecewise-cubic polynomial, respectively. Over the sampled grid-node counts, the objective decreases monotonically for each degree, and at every fixed $ k $ a higher degree gives a lower objective. The gap between $ m=1 $ and $ m=3 $ contracts as $ k $ grows. Within this tested range, local decision enrichment therefore has its largest observed effect on the coarser grids.

The horizontal line in \cref{fig:masubuchi-grid} is the grid-free LPVTools \cite{Hjartarson2015} native IQC result, $ \gamma=6.14854 $, which serves as a benchmark line. The direct curves for $ m=2 $ and $ m=3 $ fall below this reference at $ k=3 $, while the $ m=1 $ curve does so at $ k=4 $.
\begin{figure}[htbp]
    \centering
    \begin{tikzpicture}[trim axis left, trim axis right]
        \begin{axis}[
                masubuchi axis,
                width=0.90\columnwidth,
                xmin=0, xmax=10,
                xtick={0,2,...,10},
                xlabel={decision degree $ m $},
                legend columns=2,
                legend style={
                        at={(0.5,1.03)},
                        anchor=south,
                        draw=none,
                        font=\footnotesize,
                    }]
            \addplot[blue!75!black,thick,mark=*,mark size=1.5pt]
            table[x=m,y=direct,col sep=comma]
                {pics/masubuchi_k2_comparison.csv};
            \addlegendentry{Direct}
            \addplot[orange!85!black,thick,dashed,mark=triangle*,
                mark size=1.8pt]
            table[x=m,y=polya,col sep=comma]
                {pics/masubuchi_k2_comparison.csv};
            \addlegendentry{P\'olya}
            \addplot[red!75!black,thick,densely dotted,mark=square,
                mark size=3pt]
            table[x=m,y=rolmip,col sep=comma]
                {pics/masubuchi_k2_comparison.csv};
            \addlegendentry{ROLMIP}
            \addplot[green!55!black,thick,mark=diamond*,mark size=1.7pt]
            table[x=m,y=putinar,col sep=comma]
                {pics/masubuchi_k2_comparison.csv};
            \addlegendentry{Matrix-SOS (M.--L.)}
            \addplot[purple!75!black,thick,dash dot,no marks]
            table[x=m,y=lpvtools,col sep=comma]
                {pics/masubuchi_k2_comparison.csv};
            \addlegendentry{LPVTools IQC}
        \end{axis}
    \end{tikzpicture}
    \caption{Fixed-grid certificate sensitivity and interoperability across
        Lyapunov-matrix degrees.}
    \label{fig:masubuchi-certificate}
\end{figure}
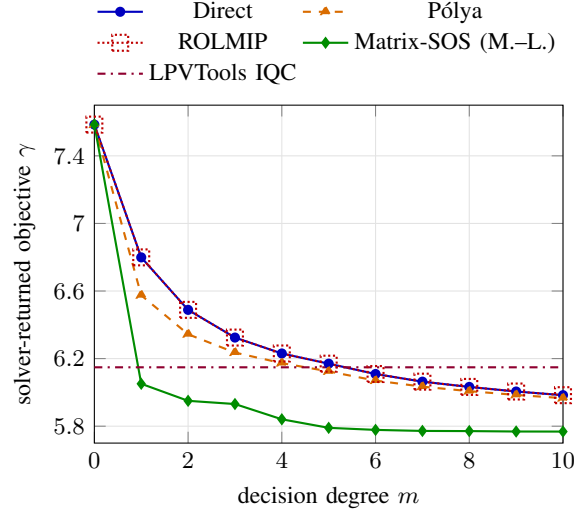

\Cref{fig:masubuchi-certificate} next fixes the grid at $ k=2 $ and compares representative coefficient and SOS certificates as the Lyapunov-matrix degree increases. The plotted GriD-LMIA certificates are the direct Bernstein coefficient test, P\'olya with $ d=1 $, and the common one-dimensional matrix Markov--Luk\'acs certificate with $ r=5 $. Dense Putinar and FullBox assemble the same interval construction, so one curve represents them. The sparse Gram paths are examined separately in the three-parameter study, where their block dimensions differ from the dense paths.

At $ m=0 $, the three certificates coincide because $ P(\rho) $ is constant. For positive degrees, the plotted matrix Markov--Luk\'acs curve is lowest, P\'olya is intermediate, and the direct Bernstein curve is highest. For this example, the trend indicates that low polynomial degree already provides substantial decision flexibility, while the interval matrix-SOS construction is less conservative than the two coefficient-based sufficient tests in the tested orders.

On the fixed single-cell grid, P\'olya elevation with increment one provides only a modest improvement over Direct, whereas the matrix Markov--Luk\'acs curve drops sharply at low degrees. Direct and P\'olya are sufficient coefficient tests for the matrix polynomial. The matrix Markov--Luk\'acs construction introduces additional Gram variables and gives a less conservative sufficient condition in this fixture. Its curve obtains most of the observed reduction with a linear or quadratic global Lyapunov polynomial and then becomes comparatively flat.

ROLMIP uses homogeneous simplex-monomial coefficients \cite{Agulhari2019}, whereas GriD-LMIA uses Bernstein coefficients. On this single-cell fixture the two coefficient sets differ by the positive binomial scaling described in \cref{sec: bernstein-basis}. The direct comparison contains 11 pairs for $ m=0,\ldots,10 $ on $ k=2 $. GriD-LMIA uses the zero-margin residual closure with SDPT3, whereas ROLMIP imposes a $10^{-8}I$ residual offset and uses SeDuMi. Every solve returns problem code zero, and the largest absolute objective discrepancy is $ 2.01\times10^{-7} $, below the prespecified $5\times10^{-4}$ agreement tolerance. This close agreement checks coefficient-translation consistency for the fixture, but the pair is not an identical numerical formulation because the residual offsets and solvers differ. The LPVTools line remains a native-IQC baseline reference.

\subsection{Three-parameter certificate and interoperability study}

The second study uses the fourth-order mass--spring model from \cite{Agulhari2019,Oliveira2008} to examine multivariate GriD-LMIA assembly, directional grid refinement, and certificate interoperability:
\begin{equation}
    \begin{aligned}
        A(\bm\rho)
          & =
        \begin{bmatrix}
            0        & 0       & 1             & 0             \\
            0        & 0       & 0             & 1             \\
            -2\rho_1 & \rho_1  & -\rho_1\rho_3 & 0             \\
            \rho_2   & -\rho_2 & 0             & -\rho_2\rho_3
        \end{bmatrix}, B(\bm\rho)
        =\begin{bmatrix}0\\0\\\rho_1\\0\end{bmatrix}, \\
        C & =\begin{bmatrix}0&1&0&0\end{bmatrix},
        \quad
        D=0 .
    \end{aligned}
    \label{eq:agulhari-plant}
\end{equation}
where the scheduling parameters and their rates satisfy
\begin{align}
    \rho_1     & \in[2/3,2],        &
    \rho_2     & \in[0.8,4/3],      &
    \rho_3     & \in[1,3],\nonumber   \\
    \dot\rho_1 & \in[-1,1],         &
    \dot\rho_2 & \in[-2/5,2/5],     &
    \dot\rho_3 & \in[-1/2,1/2].
\end{align}
The study uses the same zero-margin solver objective $ \gamma $, based on the $ \mathcal L_2 $-gain model in \cref{lemma: L2-gain}. Let $ \bm k=(k_1,k_2,k_3) $ collect the axis-node counts, so the tensor partition has $ k_1k_2k_3 $ nodes and $ (k_1-1)(k_2-1)(k_3-1) $ physical cells. To isolate directional refinement, the third-axis grid is fixed at $ \mathcal G_3=\{1,3\} $, while $ k_1 $ and $ k_2 $ vary independently. Holding the third axis fixed separates the two directional sweeps, and no inference is made from interval width alone.

The following code constructs the GriD-LMIA model for the axis-node-count vector \texttt{k}, decision degree \texttt{m} with matching Gram order $ r=m $, and common sparse Gram bandwidth \texttt{omega}.
\begin{lstlisting}[style=matlab]
k = [2 2 2]; m = 1; r = m; omega = 2;
grid = {linspace(2/3, 2, k(1)), ...
        linspace(0.8, 4/3, k(2)), ...
        linspace(1, 3, k(3))};
Afun = @(t1,t2,t3) [0 0 1 0; 0 0 0 1; ...
    -2*t1 t1 -t1*t3 0; t2 -t2 0 -t2*t3];
Bfun = @(t1,t2,t3) [0; 0; t1; 0];
A = pdmat(grid, Afun, Degree=1);
B = pdmat(grid, Bfun, Degree=1);
C = [0 1 0 0]; D = 0;
\end{lstlisting}
\noindent\emph{Listing continued: decision and certificate assembly.}
\begin{lstlisting}[style=matlab]
P = pdvar(4, grid, "symmetric", Degree=m);
gamma = sdpvar(1);
dP = rhodiff(P, [-1 1; -0.4 0.4; -0.5 0.5]);
\end{lstlisting}
\noindent\emph{Listing continued: residual, certificate, and solve.}
\begin{lstlisting}[style=matlab]
DPDLMI = [dP + P*A + A'*P, P*B, C'; ...
     B'*P, -gamma, D'; C, D, -gamma] <= 0;
% Zero-margin closure used in the reported benchmark.
% DPDLMI = DPDLMI.usePolya(3)
% DPDLMI = DPDLMI.usePutinar(r),
% DPDLMI = DPDLMI.useSpPut(omega, r),
% DPDLMI = DPDLMI.useSpBox(omega, r), or
% DPDLMI = DPDLMI.useFullBox(r).
posDefP = P >= 1e-8 * eye(4);
opts = sdpsettings('solver', 'sdpt3', 'verbose', 0);
constraints = [DPDLMI.toYalmip, posDefP.toYalmip];
sol = optimize(constraints, gamma, opts);
assert(sol.problem == 0, sol.info);
value(gamma)
\end{lstlisting}

All GriD-LMIA Direct and P\'olya results shown in \cref{fig:agulhari-grid-surface,fig: agulhari-degree} use SDPT3 \cite{Toh1999}. For the degree sweep, the P\'olya increment is fixed at $ d=3 $ and applied uniformly along all three parameter directions. The ROLMIP comparison is restricted to Direct rows, which also use SDPT3. The LPVTools line uses its native IQC analysis, while the SOS certificate rows reported separately below retain their MOSEK results \cite{MOSEKApS2019}.

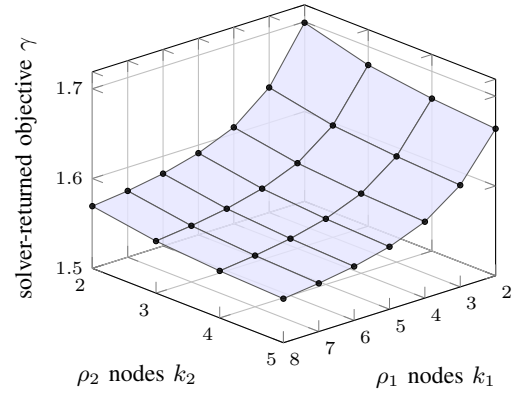
\begin{figure}[htbp]
    \centering
    \begin{tikzpicture}
        \begin{axis}[
                width=0.80\columnwidth,
                height=0.70\columnwidth,
                view={138}{27},
                xlabel={$ \rho_1 $ nodes $ k_1 $},
                ylabel={$ \rho_2 $ nodes $ k_2 $},
                zlabel={solver-returned objective $ \gamma $},
                xtick={2,...,8},
                ytick={2,...,5},
                zmin=1.5,
                ztick={1.5,1.6,1.7},
                grid=major,
                z buffer=sort,
                unbounded coords=jump,
                tick label style={font=\scriptsize},
                label style={font=\footnotesize}]
            \addplot3[
                surf,
                shader=faceted,
                faceted color=black!65,
                fill opacity=0.8,
                draw opacity=1,
                colormap={paleblue}{
                        color(0cm)=(blue!10);
                        color(1cm)=(blue!10)
                    },
                mesh/rows=4,
                mesh/ordering=x varies,
                mark=*,
                mark size=1pt,
                mark options={solid,fill=black,draw=black}]
            table[x=k1,y=k2,z=gamma,col sep=comma]
                {pics/agulhari2019_grid_surface.csv};
        \end{axis}
    \end{tikzpicture}
    \caption{Direct Bernstein coefficient certificate sensitivity with respect to $ \rho_1 $ and $ \rho_2 $ at $ m=1 $ with $ k_3=2 $.}
    \label{fig:agulhari-grid-surface}
\end{figure}

\Cref{fig:agulhari-grid-surface} fixes $ m=1 $ and $ k_3=2 $, and varies $ k_1\in\{2,\ldots,8\} $ and $ k_2\in\{2,\ldots,5\} $. All 28 SDPT3 solves return problem code zero, and the solver objective decreases monotonically along every sampled row and column. At $ k_2=2 $, increasing $ k_1 $ from 2 to 8 lowers the objective from $ 1.69895 $ to $ 1.56923 $, while at $ k_1=2 $, increasing $ k_2 $ from 2 to 5 lowers it to $ 1.66374 $. Refining both directions gives $ 1.54908 $ at $ (k_1,k_2,k_3)=(8,5,2) $. Computed from the unrounded retained records, the observed reduction is $ 0.12972 $ along the sampled $ k_1 $ edge and $ 0.03522 $ along the sampled $ k_2 $ edge. This comparison is specific to these grid paths and does not establish a general relation between interval width and refinement sensitivity.

\begin{figure}[htbp]
    \centering
    \begin{tikzpicture}[trim axis left, trim axis right]
        \begin{axis}[
                width=0.90\columnwidth,
                height=0.70\columnwidth,
                xlabel={decision degree $ m $},
                ylabel={solver-returned objective $ \gamma $},
                xmin=0, xmax=5,
                xtick={0,...,5},
                grid=major,
                tick label style={font=\scriptsize},
                label style={font=\footnotesize},
                legend style={
                        font=\scriptsize,
                        at={(0.5,1.03)},
                        anchor=south,
                        legend columns=2,
                        /tikz/every even column/.append style={column sep=3pt}
                    }]
            \addplot[blue!75!black,thick,mark=*]
            table[x=m,y=pd_direct,col sep=comma]
                {pics/agulhari2019_degree_comparison.csv};
            \addlegendentry{GriD-LMIA Direct}
            \addplot[orange!85!black,thick,dashed,mark=triangle*]
            table[x=m,y=pd_polya,col sep=comma]
                {pics/agulhari2019_degree_comparison.csv};
            \addlegendentry{GriD-LMIA P\'olya}
            \addplot[red!75!black,thick,densely dotted,mark=square]
            table[x=m,y=rolmip_direct,col sep=comma]
                {pics/agulhari2019_degree_comparison.csv};
            \addlegendentry{ROLMIP}
            \addplot[purple!75!black,thick,dash dot]
            table[x=m,y=lpvtools,col sep=comma]
                {pics/agulhari2019_degree_comparison.csv};
            \addlegendentry{LPVTools IQC}
        \end{axis}
    \end{tikzpicture}
    \caption{Solver-returned $ \gamma $ under different decision degrees and certificates on the one-cell grid $ \bm k=(2,2,2) $.}
    \label{fig: agulhari-degree}
\end{figure}
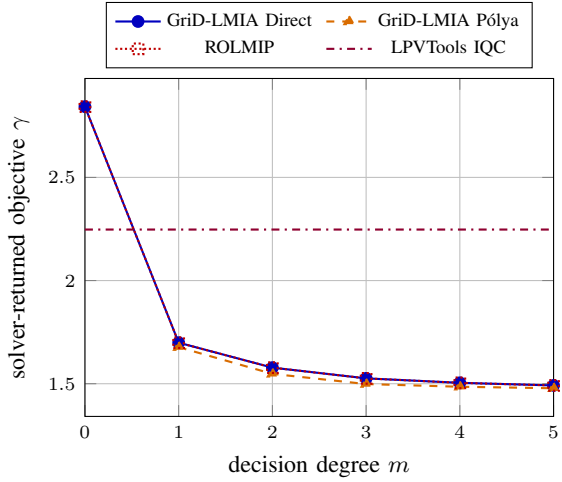

\Cref{fig: agulhari-degree} adopts the coarsest tensor grid, $ \bm k=(2,2,2) $, and increases decision degree $ m $ from zero to five. The constant-$ P $ baseline at $ m=0 $ gives the largest plotted objective, and its duplicate P\'olya point is omitted. Once the Lyapunov matrix is parameter dependent, the Direct and P\'olya objectives decrease rapidly at low degree and then flatten. With the P\'olya elevation increment fixed at $ d=3 $, P\'olya lowers the Direct objective at every $ m\geq1 $. The reduction grows from $ 0.02172 $ at $ m=1 $ to $ 0.03017 $ at $ m=2 $, then decreases to $ 0.01365 $ at $ m=5 $ as both curves flatten.

The six ROLMIP Direct rows use SDPT3 and return problem code zero. ROLMIP imposes a $10^{-8}I$ residual offset, while the plotted GriD-LMIA Direct rows use the zero-margin closure. Their markers overlap, and the largest absolute objective discrepancy is $ 9.04\times10^{-8} $, below the prespecified $5\times10^{-4}$ agreement tolerance. On this one-cell tensor grid, the homogeneous simplex-monomial and Bernstein coefficients again differ only by the positive binomial scaling. The agreement supports the coefficient translation, but these rows are not labeled as an identical numerical formulation because their residual offsets differ. The native LPVTools IQC result, $ \gamma=2.24726 $, lies between the plotted constant-$ P $ and parameter dependent-$ P $ objectives and is retained only as a baseline external reference.

The SOS-based certificates are not included in \cref{fig: agulhari-degree}. Their Gram matrices and coefficient-matching identities grow substantially with $ m $; the assembled sizes and the solver outcomes at $ m=1,2 $ are reported separately in \cref{tab:agulhari-certificates}.

\begin{table}[htbp]
    \centering
    \caption{Zero-margin one-cell outcomes. P\'olya uses $d=3$, and the sparse paths use $\omega=2$. Values require code zero; code 9 is unknown, not infeasible.}
    \label{tab:agulhari-certificates}
    \begin{tabular}{@{}lccc@{}}
        \toprule
        Certificate    & Solver & $ m=1 $ & $ m=2 $          \\
        \midrule
        Direct         & SDPT3  & 1.69895 & 1.57797          \\
        P\'olya        & SDPT3  & 1.67723 & 1.54780          \\
        Putinar        & MOSEK  & 1.66266 & unknown (code 9) \\
        SparsePutinar  & MOSEK  & 1.66266 & 1.52700          \\
        Sparse FullBox & MOSEK  & 1.66266 & 1.52457          \\
        FullBox        & MOSEK  & 1.66266 & unknown (code 9) \\
        \bottomrule
    \end{tabular}
\end{table}

\Cref{tab:agulhari-certificates} compares all six certificate paths on the same one-cell grid. The SOS rows use matching order $ r=m $ and common sparse Gram bandwidth $ \omega=2 $. \Cref{tab:agulhari-sos-scale} reports the corresponding assembled program sizes.

A separate positive-margin validation at $ m=r=2 $ and $ \omega=2 $ used $ \epsilon_P=\epsilon_L=10^{-7} $. SparsePutinar and Sparse FullBox both returned problem code zero. Their respective coefficient-identity errors were $ 8.70\times10^{-7} $ and $ 3.48\times10^{-7} $, below the $ 10^{-6} $ gate. The minimum Gram eigenvalues were $ 1.55\times10^{-4} $ and $ 7.66\times10^{-5} $, while the maximum eigenvalues of $ \mathcal F+\epsilon_LI $ were $ -2.84\times10^{-4} $ and $ -1.44\times10^{-4} $. Sampled minima of $ \lambda(P) $ were $ 0.952 $ and $ 0.964 $, and an independent plant-level reconstruction agreed within $ 8\times10^{-11} $ on a $ 5^3 $ mesh and all eight rate vertices. These checks passed the stated numerical validation gates for the strict-margin models, rather than providing a full-domain floating-point error bound. The objectives in \cref{tab:agulhari-certificates} remain the separate zero-margin comparison.

\begin{itemize}
    \item At $ m=1 $, SparsePutinar and Putinar agree exactly in the stored output because $ \omega=2 $ spans every active Gram-basis direction. All four SOS objectives round to $ 1.66266 $ at the precision shown in the table.
    \item At $ m=2 $, both sparse rows return finite objectives with problem code zero. SparsePutinar gives $ 1.52700 $, while Sparse FullBox gives $ 1.52457 $, a reduction of $ 0.00243 $ for this fixture and configuration. Putinar and FullBox both return YALMIP \texttt{problem-9}, reported as an unknown solver problem after high-memory MOSEK calls, so neither dense row has a retained objective. These resource-dependent outcomes do not establish certificate infeasibility or a solver-independent comparison.
    \item Both sparse rows retain a maximum Gram-block dimension of $ 48 $. SparsePutinar uses 160 blocks and $ 188{,}160 $ Gram scalar variables, whereas Sparse FullBox uses 216 blocks and $ 254{,}016 $ Gram scalar variables. The two dense rows have maximum dimension $ 162 $ and terminate without retained objectives.
\end{itemize}

\Cref{tab:agulhari-sos-scale} connects the certificate outcomes to the assembled finite programs. Here, Blocks is the number of positive-semidefinite Gram blocks contributed by the certificate, while Scalars counts their distinct YALMIP scalar variables and excludes the original $ P $ and $ \gamma $. Max dim.\ is the dimension of the largest Gram block, and Identities is the number of scalar coefficient-matching equalities. SparsePutinar retains the Putinar generator terms with narrower Gram bands, while Sparse FullBox retains the FullBox traversal and decomposes its Gram representation into overlapping tensor-window blocks.

\begin{table}[htbp]
    \centering
    \caption{SOS assembly dimensions, with $ \omega=2 $ for the sparse paths.}
    \label{tab:agulhari-sos-scale}
    \begingroup
    \footnotesize
    \setlength{\tabcolsep}{1.9pt}
    \begin{tabular}{@{}clrrrr@{}}
        \toprule
        $ m $ & Certificate    & Blocks & Scalars     & Max dim. & Identities \\
        \midrule
        1     & Putinar        & 32     & $16{,}608$  & 48       & $4{,}536$  \\
        1     & SparsePutinar  & 32     & $16{,}608$  & 48       & $4{,}536$  \\
        1     & Sparse FullBox & 64     & $18{,}648$  & 48       & $4{,}536$  \\
        1     & FullBox        & 64     & $18{,}648$  & 48       & $4{,}536$  \\
        2     & Putinar        & 32     & $246{,}888$ & 162      & $21{,}000$ \\
        2     & SparsePutinar  & 160    & $188{,}160$ & 48       & $21{,}000$ \\
        2     & Sparse FullBox & 216    & $254{,}016$ & 48       & $21{,}000$ \\
        2     & FullBox        & 64     & $319{,}368$ & 162      & $21{,}000$ \\
        \bottomrule
    \end{tabular}
    \endgroup
\end{table}

The separate condition $ P(\bm\rho)\succ0 $ adds $ (m+1)^3 $ positive-semidefinite blocks on this one-cell grid: eight blocks at $ m=1 $ and 27 at $ m=2 $.  Each is only $ 4\times4 $, adding respectively $ 80 $ and $ 270 $ scalarized symmetric entries. This common addition is small relative to the Gram systems in the table and therefore does not explain the different solver outcomes.

From $ m=1 $ to $ m=2 $, the maximum dense Gram dimension for Putinar and FullBox increases from $ 48 $ to $ 162 $, while the scalar-variable counts in \cref{tab:agulhari-sos-scale} increase by roughly an order of magnitude. SparsePutinar and Sparse FullBox retain a maximum block dimension of $ 48 $ by introducing overlapping positive-semidefinite blocks. At $ m=2 $, SparsePutinar uses fewer Gram scalars than Sparse FullBox but returns a slightly higher objective. These dimensions and outcomes describe different assembled cones and do not establish a backend-independent runtime or numerical-stability ranking.

For this multi-parameter case, grid refinement and a higher decision degree lower the Direct objectives. Both sparse Gram decompositions keep the largest block at dimension $ 48 $ and return finite solver objectives at $ m=2 $, with Sparse FullBox lower by $ 0.00243 $ in this run. Thus, the certificate choice changes both the representable cone and the assembled Gram structure, while the observed solver outcomes remain specific to the reported environment and backend settings.

\section{Conclusion}
\label{sec:conclusion}

GriD-LMIA constructs inspectable finite semidefinite models for affine DPD-LMIs over hyper-rectangular parameter and rate domains. It combines shared Bernstein decision coefficients, rate-vertex differentiation, adjacent-cell interface conditions, and interchangeable sufficient-certificate assemblers. Derivative-free robust PD-LMIs are included as a special case. The one-parameter study examines degree, grid, and representative certificate choices and checks coefficient-translation consistency against a closely aligned ROLMIP formulation. The three-parameter study examines directional grid and degree choices together with the assembled scale and solver outcomes of all four SOS interfaces. Both studies keep baseline native LPVTools analyses separate from package ranking.

GriD-LMIA is appropriate when users need direct control over axis grids, Bernstein decision degree $ m $, P\'olya increment $ d $, common sparse Gram bandwidth $ \omega $, or absolute Gram order $ r $. A failed fixed-order certificate remains inconclusive about the continuous problem, and products of two decision-dependent factors remain outside the affine modeling contract.

\appendices
\crefalias{section}{appendix}
\section{Bernstein Coefficient Algebra}
\label{app:bernstein-algebra}

This appendix records the exact coefficient operations used by the modeling layer. For a scalar degree $ m $, the coefficient labels traverse $ \{0,\ldots,m\}^{\ell} $.  On each physical cell,
\begin{equation}
    P^{(\mathbf c)}(\bm\alpha)
    =\sum_{\mathbf i\in\{0,\ldots,m\}^{\ell}}
    B_{\mathbf i}^{m}(\bm\alpha)P^{(\mathbf c)}[\mathbf i].
    \label{eq:tensor-bernstein-expansion}
\end{equation}
Each operation below is stated for the general $ \ell $-parameter representation before its scalar specialization.  The scalar formulas and MATLAB snippets are one-parameter demos chosen only for compact display.

\subsection{Lossless degree elevation}

Consider a degree-$ m $ polynomial on physical cell $ \mathcal H_{\mathbf c} $. Uniform elevation to $ M\geq m $ in all $ \ell $ directions is:
\begin{equation}
    \widehat C^{(\mathbf c)}[\mathbf k]
    =\sum_{\mathbf i\in\{0,\ldots,m\}^{\ell}}
    \left[
        \prod_{s=1}^{\ell}
        \frac{\binom{m}{i_s}\binom{M-m}{k_s-i_s}}
        {\binom{M}{k_s}}
        \right]C^{(\mathbf c)}[\mathbf i].
    \label{eq:direct-elevation}
\end{equation}
for every $ \mathbf k\in\{0,\ldots,M\}^{\ell} $, where the negative terms in the binomial coefficients are understood to be zero. Equivalently, this can be written as a $ \ell $-dimensional discrete convolution with kernel:
\begin{equation*}
    \mathcal E_{M-m}[\mathbf j]
    :=\prod_{s=1}^{\ell}\binom{M-m}{j_s}, \quad \mathbf j\in\{0,\ldots,M-m\}^{\ell}.
\end{equation*}
Here $ \ast_\ell $ denotes the $ \ell $-dimensional discrete convolution. \cref{eq:direct-elevation} can be expressed as
\begin{multline}
    \left\{
    \prod_{s=1}^{\ell}\binom{M}{k_s}
    \widehat C^{(\mathbf c)}[\mathbf k]
    \right\}_{\mathbf k}
    =\\
    \left\{
    \prod_{s=1}^{\ell}\binom{m}{i_s}
    C^{(\mathbf c)}[\mathbf i]
    \right\}_{\mathbf i}
    \ast_{\ell}\left\{
    \mathcal E_{M-m}[\mathbf j]
    \right\}_{\mathbf j}.
    \label{eq: elevation-convolution}
\end{multline}

The elevation is deterministic as the convolution \cref{eq: elevation-convolution} uniquely calculates the degree-$ M $ representation from the degree-$ m $ \cref{eq:tensor-bernstein-expansion} and the degree-$ (M-m) $ convolution kernel.

\begin{lstlisting}[style=matlab]
p = pdmat({[0 1]}, {1, 2, 6}, Degree=2);
pElevated = p.elevate(1);
\end{lstlisting}

\subsection{Addition and subtraction}

Let $ P $ and $ Q $ have degrees $ m $ and $ n $.  Elevate their coefficient arrays to $ M=\max(m,n) $ using \eqref{eq:direct-elevation}, and denote the results by $ \widehat P $ and $ \widehat Q $.  For every physical cell $ \mathbf c $ and $ \mathbf i\in\{0,\ldots,M\}^{\ell} $,
\begin{equation}
    (P\mathbin{\pm}Q)^{(\mathbf c)}[\mathbf i]
    =\widehat P^{(\mathbf c)}[\mathbf i]
    \mathbin{\pm}\widehat Q^{(\mathbf c)}[\mathbf i].
\end{equation}
GriD-LMIA performs this tensor-degree alignment automatically.  A one-parameter demo is
\begin{lstlisting}[style=matlab]
p = pdmat({[0 1]}, {1, 2, 6}, Degree=2);
q = pdmat({[0 1]}, {2, 4}, Degree=1);
s = p + q;
degree = s.Degree;                    % 2
coefficients = cell2mat(s.coeffs(1)); % [3 5 10]
T = s.bernTable("oneLine");
\end{lstlisting}
The same common-degree step precedes subtraction.  When the operands use different compatible grids, the algebra represents them on their common refinement.

\subsection{Matrix multiplication}

Let $ P^{(\mathbf c)} $ and $ Q^{(\mathbf c)} $ have degrees $ m $ and $ n $:
\begin{equation}
    \begin{aligned}
        P^{(\mathbf c)}
         & =
        \sum_{\mathbf i\in\{0,\ldots,m\}^{\ell}}
        B_{\mathbf i}^{m}P^{(\mathbf c)}[\mathbf i], \\
        Q^{(\mathbf c)}
         & =
        \sum_{\mathbf j\in\{0,\ldots,n\}^{\ell}}
        B_{\mathbf j}^{n}Q^{(\mathbf c)}[\mathbf j].
    \end{aligned}
\end{equation}
Applying the Bernstein product identity independently in every parameter direction gives, for $ \mathbf k\in\{0,\ldots,m+n\}^{\ell} $,
\begin{equation}
    \begin{aligned}
        (PQ)^{(\mathbf c)}[\mathbf k]
        ={} & \sum_{\substack{
                      \mathbf i\in\{0,\ldots,m\}^{\ell}\\
                      \mathbf j\in\{0,\ldots,n\}^{\ell},\
                      \mathbf i+\mathbf j=\mathbf k}}
        P^{(\mathbf c)}[\mathbf i]Q^{(\mathbf c)}[\mathbf j]
        \prod_{s=1}^{\ell}
        \frac{\binom{m}{i_s}\binom{n}{j_s}}
        {\binom{m+n}{k_s}}.
    \end{aligned}
    \label{eq:bernstein-product}
\end{equation}
where the negative terms in the binomial coefficients are understood to be zero. \cref{eq:bernstein-product} also admits the following division-free $ \ell $-dimensional discrete convolution form:
\begin{gather*}
    \left\{
    \prod_{s=1}^{\ell}\binom{m+n}{k_s}
    (PQ)^{(\mathbf c)}[\mathbf k]
    \right\}_{\mathbf k}
    =\\
    \left\{
    \prod_{s=1}^{\ell}\binom{m}{i_s}
    P^{(\mathbf c)}[\mathbf i]
    \right\}_{\mathbf i}
    \ast_{\ell}
    \left\{
    \prod_{s=1}^{\ell}\binom{n}{j_s}
    Q^{(\mathbf c)}[\mathbf j]
    \right\}_{\mathbf j}.
\end{gather*}
Because $ P $ and $ Q $ are matrices, the ordered matrix multiplication is preserved. If at most one factor contains decision variables, the product remains affine in those decisions. Products of two decision-dependent factors remain outside the affine modeling contract.

For the one-parameter scalar rows $ P=[1,2,6] $ and $ Q=[2,4] $, the multiplication gives $ PQ=[2,4,28/3,24] $:
\begin{lstlisting}[style=matlab]
P = pdmat({[0 1]}, {1, 2, 6}, Degree=2);
Q = pdmat({[0 1]}, {2, 4}, Degree=1);
R = P * Q;
productDegree = R.Degree;             % 3
productCoefficients = cell2mat(R.coeffs(1));
% productCoefficients = [2 4 28/3 24]
\end{lstlisting}

\subsection{Differentiation and rate vertices}

Define the indicator vector $ \mathbf e_s\in\mathbb{R}^{\ell} $ by the $ s $th entry to be one and zero elsewhere, and $ h_s^{\mathbf c} $ as defined in \cref{eq: local-coordinate}. For $ m\geq1 $, the general tensor derivative is
\begin{equation}
    \begin{aligned}
        \frac{\partial P^{(\mathbf c)}}{\partial\rho_s}
        ={} & \frac{m}{h_s^{\mathbf c}}
        \sum_{\substack{\mathbf i\in\{0,\ldots,m\}^\ell\\i_s\leq m-1}}
        \left(P^{(\mathbf c)}[\mathbf i+\mathbf e_s]
        -P^{(\mathbf c)}[\mathbf i]\right)     \\
            & {}\times B_{i_s}^{m-1}(\alpha_s)
        \prod_{\substack{t=1\\t\neq s}}^\ell B_{i_t}^{m}(\alpha_t).
    \end{aligned}
    \label{eq:bernstein-partial}
\end{equation}
where the summation index satisfies $ 0\leq i_s\leq m-1 $ while $ 0\leq i_t\leq m $ for $ t\neq s $. For $ m=0 $, every partial derivative is zero.

This partial derivative has degree $ m-1 $ in direction $ s $ and degree $ m $ in every other direction.  For the one-parameter demo, with $ h_c=\rho_1^{(c+1)}-\rho_1^{(c)} $, \eqref{eq:bernstein-partial} becomes
\begin{equation}
    \frac{\partial P^{(c)}}{\partial\rho_1}
    =\frac{m}{h_c}\sum_{i=0}^{m-1}
    \left(P^{(c)}[i+1]-P^{(c)}[i]\right)B_i^{m-1}(\alpha).
    \label{eq:bernstein-derivative-1d}
\end{equation}
Thus $ [1,2,6] $ on $ [0,2] $ differentiates to $ [1,4] $.  For $ \ell>1 $, the multivariate partials initially have different tensor degrees. \texttt{rhodiff} elevates each partial only in its differentiated direction back to the common degree $ m $, forms
\begin{equation}
    \dot P(\bm\rho)=\sum_{s=1}^{\ell}
    \frac{\partial P}{\partial\rho_s}(\bm\rho)\dot\rho_s,
\end{equation}
and stores one coefficient row for each of the $ 2^\ell $ elements of $ \mathcal V_{\mathcal R} $.  When $ \ell=1 $, no cross-direction alignment is needed, so the result retains degree $ m-1 $.  The corresponding one-parameter demo is
\begin{lstlisting}[style=matlab]
yalmip('clear')
P = pdvar(1, {[0 2]}, Degree=2, ...
    RateBounds=[-1 1]);
dP = rhodiff(P);
T = dP.bernTable("oneLine");
% One degree-one row per rate vertex.
\end{lstlisting}

\bibliographystyle{IEEEtran}
\bibliography{references}

\begin{IEEEbiography}[{\includegraphics[width=1in,height=1.25in,clip,keepaspectratio]{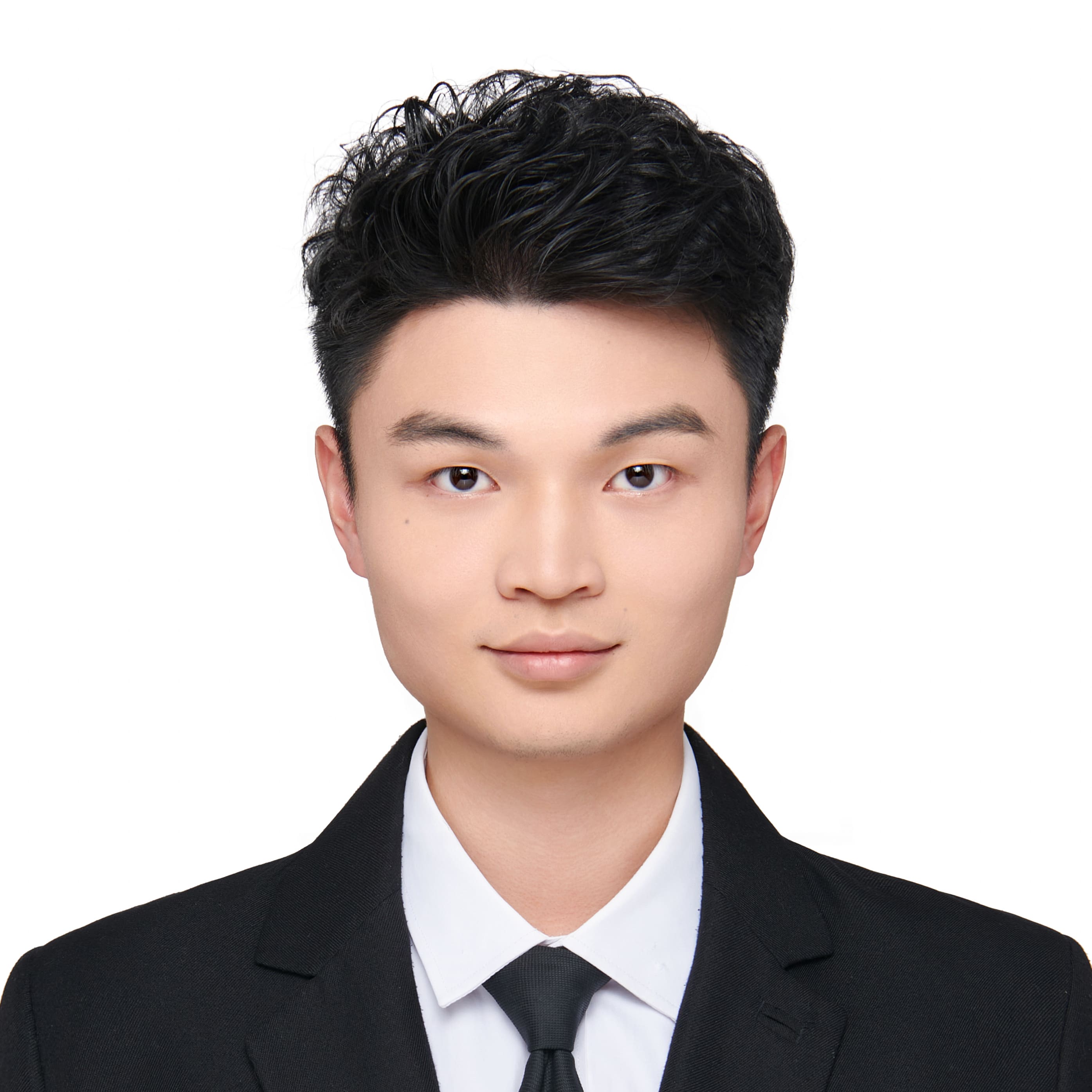}}]{Yicheng Xu}
received the B.S. degree in automation from Southeast University, Nanjing, China, in 2020 and the M.S. degree in mechanical engineering from the University of California, Irvine, in 2021. He has been a Ph.D. student since 2021. His research interests include distributed control of multiagent systems, event-triggered control, and anti-windup control.
\end{IEEEbiography}

\begin{IEEEbiography}[{\includegraphics[width=1in,height=1.25in,clip,keepaspectratio]{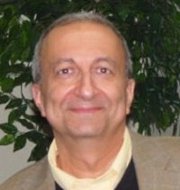}}]{Faryar Jabbari}
is on the faculty of the Mechanical and Aerospace Engineering Department at UCI. His research is in control theory and its applications. He has served as an associate editor for Automatica and IEEE Transactions on Automatic Control, the program chair for ACC-11 and CDC-09, and the general chair for CDC-14.
\end{IEEEbiography}

\end{document}